%% file: manuscript.tex
\documentclass{article}
\usepackage{arXiv}            

\usepackage{amsthm}

\usepackage{booktabs}            
\usepackage{multirow}            
\usepackage{amsfonts}            
\usepackage{graphicx}            
\usepackage{duckuments}          
\usepackage{amsmath,amssymb,amsfonts}
\usepackage{graphicx}
\usepackage{textcomp}
\usepackage{xcolor}
\usepackage{contour}

\usepackage[T1]{fontenc}    
\usepackage{url}            
\usepackage{booktabs}       
\usepackage{amsfonts}       
\usepackage{nicefrac}       
\usepackage{microtype}      
\usepackage[center]{subfigure} 
\usepackage{graphicx}  
\usepackage{amsmath}
\usepackage{xspace}
\usepackage{amssymb}
\usepackage{pifont}
\usepackage{stmaryrd} 

\usepackage{chngcntr}
\usepackage{txfonts}

\usepackage{caption}
\usepackage{enumitem}
\setitemize{noitemsep,topsep=0pt,parsep=0pt,partopsep=0pt}
\usepackage{tabularx}
\usepackage{rotating}
\usepackage{tablefootnote}
\usepackage{threeparttable}
\usepackage{xcolor}
\usepackage{wasysym}
\usepackage{multicol}
\usepackage{multirow}
\usepackage{mathtools}
\usepackage{esvect}
\usepackage{pdflscape}
\usepackage{makecell}

\usepackage{algorithm}
\usepackage{algpseudocode}
\usepackage{setspace}
\usepackage{tikz}

\usepackage[numbers,compress,sort]{natbib}

\input{000dfn}

\title{
Personalized Ranking on Cascading Behavior Graphs for Accurate Multi-Behavior Recommendation
}

\author[G. Ko et al.]{%
Geonwoo Ko\thanks{Two first authors have contributed equally to this work.}\footnotemark[1]\\
\institute{Soongsil University}\\
\email{geonwooko@soongsil.ac.kr}\And
Minseo Jeon\footnotemark[1]\\
\institute{Soongsil University}\\
\email{minseojeon@soongsil.ac.kr}\And
Jinhong Jung\thanks{Corresponding author.}\footnotemark[2]\\
\institute{Soongsil University}\\
\email{jinhong@ssu.ac.kr}
}

\begin{document}

\maketitle

\input{001abstract}

\input{010introduction}

\input{020related}

\input{030preliminary}

\input{040method}

\input{050experiment}

\input{060conclusion}

\bibliographystyle{unsrtnat}
\bibliography{reference}

\appendix
\input{070appendix}

\end{document}

%% file: 000dfn.tex
\newtheorem{problem}{Problem}

\newtheorem{definition}{Definition}
\newtheorem{theorem}{Theorem}

\newcommand{\smallsection}[1]{\noindent{\smash{\textbf{#1.}}}}

\newcommand{\vect}[1]{\mathbf{#1}}
\newcommand{\mat}[1]{\mathbf{#1}}

\newcommand{\set}[1]{\mathcal{#1}}

\newcommand{\users}{\set{U}}
\newcommand{\edges}{\set{E}}
\newcommand{\Ub}{\users_{b}}
\newcommand{\items}{\set{I}}
\newcommand{\Ib}{\items_{b}}
\newcommand{\Ibp}{\items_{b'}}
\newcommand{\Ubp}{\users_{b'}}

\newcommand{\B}{\set{B}}

\newcommand{\C}{\set{C}}

\newcommand{\Eb}{\set{E}_{b}} 

\newcommand{\graph}{\set{G}}
\newcommand{\Gb}{\mathcal{G}_{b}} 
\newcommand{\Gind}[1]{\mathcal{G}_{#1}} 
\newcommand{\GC}{\mathcal{G}_{\C}} 

\newcommand{\Ab}{\mat{A}_{b}} 
\newcommand{\tAb}{\mat{\tilde{A}}_{b}} 

\newcommand{\Abnorm}{\tilde{\mat{A}}_{b}}

\renewcommand{\L}{\mat{L}}
\newcommand{\Linv}{\mat{\mathcal{L}}}
\newcommand{\Linvprod}[2]{\Linv_{b_{#1}\leftsquigarrow b_{#2}}}

\newcommand{\NIb}[1]{\set{N}_{\Ib}(#1)}
\newcommand{\NUb}[1]{\set{N}_{\Ub}(#1)}

\newcommand{\Dub}{\mat{D}_{\Ub}}
\newcommand{\Dib}{\mat{D}_{\Ib}}

\newcommand{\taobao}{\texttt{Taobao}\xspace}
\newcommand{\tmall}{\texttt{Tmall}\xspace}

\newcommand{\tenrec}{\texttt{Tenrec}\xspace}

\newcommand{\method}{\textsc{CascadingRank}\xspace} 
\newcommand{\methodcol}{\textsc{CascadingRank-col}}
\newcommand{\methodsym}{\textsc{CascadingRank-sym}}

\newcommand{\qb}{\vect{q}_{b}}
\newcommand{\qu}{\vect{q}_{\users}}
\newcommand{\qi}{\vect{q}_{\items}}
\newcommand{\qub}{\vect{q}_{\Ub}}
\newcommand{\qib}{\vect{q}_{\Ib}}

\newcommand{\qindhat}[1]{{\hat{\vect{q}}_{b_{#1}}}}

\newcommand{\rb}{\vect{r}_{b}}
\newcommand{\rbp}{\vect{r}_{b'}}
\newcommand{\ru}{\vect{r}_{\users}}
\newcommand{\ri}{\vect{r}_{\items}}
\newcommand{\rub}{\vect{r}_{\Ub}}
\newcommand{\rib}{\vect{r}_{\Ib}}
\newcommand{\rubp}{\vect{r}_{\Ubp}}
\newcommand{\ribp}{\vect{r}_{\Ibp}}

\newcommand{\wui}{\tilde{w}_{ui}}
\newcommand{\wiu}{\tilde{w}_{iu}}
\newcommand{\wuib}{\tilde{w}_{b}(u, i)}
\newcommand{\wiub}{\tilde{w}_{b}(i, u)}

\let\oldnl\nl
\newcommand{\nonl}{\renewcommand{\nl}{\let\nl\oldnl}}

\setlist[itemize]{topsep=1pt}
\setlist[enumerate]{topsep=1pt}


%% file: 001abstract.tex
\begin{abstract}
Multi-behavior recommendation predicts items a user may purchase by analyzing diverse behaviors like viewing, adding to a cart, and purchasing. Existing methods fall into two categories: representation learning and graph ranking. Representation learning generates user and item embeddings to capture latent interaction patterns, leveraging multi-behavior properties for better generalization. However, these methods often suffer from over-smoothing and bias toward frequent interactions, limiting their expressiveness. Graph ranking methods, on the other hand, directly compute personalized ranking scores, capturing user preferences more effectively. Despite their potential, graph ranking approaches have been primarily explored in single-behavior settings and remain underutilized for multi-behavior recommendation.
In this paper, we propose \method, a novel graph ranking method for multi-behavior recommendation. It models the natural sequence of user behaviors (e.g., viewing, adding to cart, and purchasing) through a cascading behavior graph. An iterative algorithm computes ranking scores, ensuring smoothness, query fitting, and cascading alignment. Experiments on three real-world datasets demonstrate that \method outperforms state-of-the-art methods, with up to 9.56\% and 7.16\% improvements in HR@10 and NDCG@10, respectively. Furthermore, we provide theoretical analysis highlighting its effectiveness, convergence, and scalability, showcasing the advantages of graph ranking in multi-behavior recommendation.
\end{abstract}

%% file: 010introduction.tex
\section{Introduction}
\label{sec:introduction}
Given graphs of multi-behavior interactions, how can we accurately rank items that a user is likely to purchase?
Multi-behavior recommendation~\cite{LoniPLH16} aims to recommend items to be purchased by a specific user, analyzing plentiful interactions across various user behaviors. 
Unlike early recommender systems~\cite{KorenBV09,RendleRFGS12,HeHDWLZW20} that rely on single-behavior interactions, the multi-behavior recommendation can more precisely capture user preferences for the target behavior (e.g., \textit{purchase}) by leveraging rich information from auxiliary behaviors (e.g., \textit{viewing} and \textit{adding to cart}). 
As a result, it has recently gained significant attention from data mining communities~\cite{JinJGHJL20,LeeLKSJ24xxaw,LiLCYLLD24,GaoGHGCFLCJ19,ChengCHLZGP23fqvn,YanYCSSP23} across various industrial domains, including streaming services, e-commerce, social media, and content aggregation.

To provide a recommendation list for a user, it is essential to calculate personalized ranking scores on items w.r.t. that user, particularly by analyzing user-item interactions, represented as a bipartite graph between users and items, based on the assumption\footnote{Collaborative filtering assumption indicates that users with similar preferences in the past will continue to have similar preferences in the future.} underlying collaborative filtering (CF)~\cite{SchaferFHS07}.
Existing methods for obtaining personalized rankings on graph data fall into two main categories: 
\textit{representation learning methods} and \textit{graph ranking methods}.
The former focuses on extracting representation vectors of users and items from the data, which are then used to predict the scores, while the latter aims to directly compute the scores by analyzing the relationships between users and items on graphs.

Many researchers have recently made tremendous efforts to develop representation learning methods within the CF framework for multi-behavior recommendation.
Their previous studies~\cite{LeeLKSJ24xxaw, ChengCHLZGP23fqvn, MengZGGLZZTLZ23, YanYCSSP23, MengMZYZL23} extend matrix factorization (MF)~\cite{KorenBV09} and graph neural networks (GNNs)~\cite{WangWHWFC19,HeHDWLZW20} to encode user and item embeddings for multi-behavior interactions, optimizing them to rank positive items higher than negative ones~\cite{RendleRFGS12}.
In particular, this approach allows models to 1) easily incorporate the inductive bias of multi-behavior (e.g., a natural sequence of user behaviors, such as viewing, adding to cart, and purchasing) in their encoding step~\cite{ChengCHLZGP23fqvn, MengMZYZL23, LiuXWY00024}, and 2) leverage these embeddings for multi-task learning to enhance generalization power~\cite{MengZGGLZZTLZ23,ZhangBCSYGWHH24,GaoGHGCFLCJ19}.
However, their recommendation quality remains limited because these methods are prone to producing over-smoothed embeddings\footnote{Over-smoothing refers to the phenomenon where node representations become too similar to each other.} under the assumption of CF~\cite{LiuGJ20,XiaHSX23,LiLCYLLD24}, particularly when utilizing GNNs, which limits their expressiveness.
Furthermore, the optimization across all users often prioritizes  users or items with a large number of multi-behavior interactions~\cite{YinCLYC12,abs-2410-04830}, limiting the ability to discover items in the long-tail distribution.


Unlike representation learning, graph ranking methods directly generate ranking scores for items with respect to a specific user.
Traditional methods~\cite{PanPYFD04,BrinB98,DengDLK09kqpg,HeHGKW17} for graph ranking have focused on single-behavior graphs and smooth the ranking scores of neighbors\footnote{This follows the smoothness assumption~\cite{ZhouZBLWS03,AgarwalA06}, which states that the ranking of a node is influenced by the rankings of its neighbors.}, while incorporating information from a querying user (e.g., interacted items) to generate ranking scores, each with its own unique design.
In particular, the smoothness property helps these methods effectively identify similar nodes to the querying node, whose performance in collaborative filtering with implicit feedback has been empirically shown to outperform the representation learning methods in previous studies~\cite{ParkPJK17,HeHGKW17,GoriGP07,LeeSKLL11}.
However, relying on just one behavior (e.g., view) may fail to accurately capture users’ genuine interest in the target behavior (e.g., purchase), and the graph ranking approach for multi-behavior recommendation remains underexplored.

In this work, we explore the graph ranking approach for multi-behavior recommendation, and propose \method, a novel personalized graph ranking method tailored for it.
To leverage the semantics of multi-behaviors of users, we first construct a \textit{cascading behavior graph} by linking behavior bipartite graphs between users and items in the order of a natural (or cascading) sequence of behaviors (e.g., \texttt{view} $\rightarrow$ \texttt{cart} $\rightarrow$ \texttt{buy}), where later behaviors exhibit stronger user preferences for the target behavior compared to earlier ones.
We then design our ranking model and its iterative algorithm that produce ranking scores along the cascading behavior graph, ensuring the smoothness for CF and fitting the query information on the current behavior while incorporating scores from the previous behavior.
Through this process, our ranking scores precisely capture users’ preferences for the target behavior while leveraging the graph structure formed by multi-behavior interactions.
%
Our main contributions are summarized as follows:
\begin{itemize}[leftmargin=9mm,noitemsep]
    \item {
        We propose \method, a new ranking model on a cascading behavior graph for accurate multi-behavior recommendation, and develop an iterative algorithm for computing our ranking scores.
    }
    \item {
        We theoretically analyze our ranking model and algorithm in terms of cascading effect, convergence, and scalability.
    }
    \item {
        We conduct extensive experiments on three real-world datasets for multi-behavior recommendation, comparing \method with state-of-the-art ranking and representation learning methods.
    }
\end{itemize}

Our research findings in this paper reveal the following strengths of \method:
\begin{itemize}[leftmargin=9mm,noitemsep]
    \item{
        \textbf{Accurate}:
        Our method provides accurate multi-behavior recommendation, achieving higher accuracy than competitors in HR@$k$ and NDCG@$k$ across various values of $k$, with improvements of up to 9.56\% in HR@10 and 7.16\% in NDCG@10.
    }
    \item{
        \textbf{Reliable}:
        The convergence of our iterative algorithm is guaranteed, producing reliable rankings.
    }
    \item{
        \textbf{Scalable}:
        The running time of our algorithm scales linearly with the number of interactions.
    }
\end{itemize}

\def\UrlFont{\ttfamily}
For reproducibility, the code and the datasets are publicly available
at \url{https://github.com/geonwooko/CascadingRank}.
The rest of the paper is organized as follows. 
We review previous methods in Section~\ref{sec:related}.
After introducing preliminaries in Section~\ref{sec:preliminaires}, we describe our proposed \method in Section~\ref{sec:proposed}.
In Sections~\ref{sec:experiments}~and~\ref{sec:conclusion}, we present our experimental results and conclusions, respectively.

%% file: 020related.tex
\section{Related Work}
\label{sec:related}
In this section, we review previous studies for ranking in recommendation systems: 1) representation learning methods and 2) graph ranking methods.

\subsection{Representation Learning Methods for Recommendation}
\label{sec:related:rl}
Representation learning extracts latent representation vectors (or embeddings) of users and items from user-item interactions, using these embeddings to yield scores for recommendations.
Early research on single-behavior recommendation modeled user-item interactions based on a specific behavior, such as viewing or rating, using Matrix Factorization (MF) or Graph Neural Networks (GNNs) within the framework of collaborative filtering.
MF~\cite{RendleRFGS12} decomposes a user-item interaction matrix into low-dimensional user and item representations, optimizing Bayesian Personalized Ranking (BPR)~\cite{RendleRFGS12} loss to rank a user’s consumed items higher than unconsumed ones.
GNNs~\cite{KipfKW16} learn the embeddings from a bipartite graph of the user-item interactions through message-passing, which aggregates embeddings from neighboring nodes.
For example, NGCF~\cite{WangWHWFC19} employs GNNs to extract user and item embeddings, effectively modeling high-order connectivity through multiple rounds of message-passing. 
LightGCN~\cite{HeHDWLZW20} simplifies the message-passing mechanism used in NGCF by removing several learning components (e.g., feature transformation and non-linear activation) to reduce overfitting and improve  efficiency.

However, relying solely on a single behavior is limited in effectively capturing user preferences, especially in multi-behavior recommendation, because such interaction data is insufficient, and it does not account for the distinct semantics associated with each behavior.
To fully leverage rich interactions from  multiple behaviors, recent studies~\cite{MengZGGLZZTLZ23} have extended GNN-based approaches to model multi-behavior interactions.
For example, MBGCN~\cite{JinJGHJL20} constructs a unified graph from multi-behavior interactions, assigning a specific behavior type to each edge, and then uses heterogeneous GNNs to learn representations from the graph.
MB-HGCN~\cite{YanYCSSP23} hierarchically learns representations by using GNNs, starting from a unified graph and progressing to behavior-specific graphs.
MuLe~\cite{LeeLKSJ24xxaw} proposes a multi-grained graph learning framework to capture diverse aspects between behaviors, exploiting graph attention to denoise uncertain auxiliary interactions.

Several researchers~\cite{LiuXWY00024, YanCGSLSL24, ChengCHLZGP23fqvn} have exploited \textit{cascading pattern} as an inductive bias, a natural sequence of user behaviors (e.g., a user first views an item, adds it to a cart, and then makes a purchase), where the later a behavior occurs in the sequence, the greater its influence on the target behavior.
Specifically, MB-CGCN~\cite{ChengCHLZGP23fqvn} directly leverages the cascading pattern in its representation learning with GNNs, sequentially refining user and item embeddings by learning from each behavior graph in the order of the sequence.
PKEF~\cite{MengMZYZL23} enhances the cascading graph learning process by incorporating signals from each behavior graph.
HEC-GCN~\cite{yin2024hecgcn} learns the structure of behavior-specific hypergraphs of users or items, and performs a cascading learning on those hypergraphs at the coarse-grained level, combining this with the cascading learning of behavior-specific graphs at the fine-grained level.

However, the existing methods are likely to produce over-smoothed embeddings due to the assumptions in collaborative filtering and GNNs~\cite{LiuGJ20,XiaHSX23,LiLCYLLD24} that encourage the embeddings of connected nodes to be similar.
Furthermore, they optimize a likelihood across all users, which causes a bias toward heavy users with a large number of interactions in their learning~\cite{YinCLYC12,abs-2410-04830}.
Due to these issues, the representation learning methods produce embeddings with limited expressiveness, hindering their ability to precisely capture user preferences from multi-behavior interactions.

\subsection{Graph Ranking Methods for Recommendation}
\label{sec:related:graph}
Graph ranking aims to directly produce ranking scores on items by analyzing the relationships between nodes (e.g., users and items) within a graph, and it has been widely used for recommending items to querying users on a graph of user-item interactions~\cite{GoriGP07,DengDLK09kqpg,HeHGKW17}, as well as in other domains such as search~\cite{KleinbergK99,BrinB98}, social network analysis~\cite{HaveliwalaH02, JungJJSK16,JungJK20}, and more~\cite{ValdeolivasVTNPOLCRB19,CowenCIRS17,JinJK19}.

Numerous ranking models~\cite{PanPYFD04,KleinbergK99,DengDLK09kqpg,HeHGKW17,MaMGZ08,BrinB98,HaveliwalaH02} for graphs have been proposed, each with its own assumptions about ranking scores, most of which focus on single-behavior interactions (e.g., clicks or ratings) for recommendation~\cite{GoriGP07,ParkPJK17,HeHGKW17,NikolakopoulosNK19,TianJ13}.
%
For example, random walk with restart (RWR)~\cite{PanPYFD04, MaMGZ08,ShinJSK15,JungPSK17,LeeLJ22}, a variant of PageRank~\cite{BrinB98}, exploits a random surfer that stochastically either performs a random walk to neighboring nodes or restarts at a querying node, yielding personalized ranking scores, with more frequently visited nodes being ranked higher, across various types of graphs~\cite{JungJJSK16,JungJPK21,ChunLSJ24}.
Hyperlink-induced topic search (HITS)~\cite{KleinbergK99} assumes that good hubs links to good authorities, and iteratively updates hub and authority scores for each node based on the graph structure. 

Several methods have been developed specifically for bipartite graphs (e.g., user-item graphs).
For example, Deng et al.~\cite{DengDLK09kqpg} developed Co-HITS, a generalized method combining HITS and RWR for bipartite graphs, controlled by flexible personalized parameters, which results in scores suited for the input graph.
BiRank~\cite{HeHGKW17} follows smoothness convention~\cite{ZhouZBLWS03,AgarwalA06} on bipartite graphs, i.e., a node on one side should be ranked higher if it is linked to higher-ranked nodes on the other side, while incorporating prior information (e.g., consumed items) and diminishing the influence of high-degree nodes during the ranking process.

However, the previous methods have dealt with ranking on graphs of single-behavior interactions, and relying on a single behavior (e.g., clicks) struggles to capture the genuine intent of users regarding the target behavior (e.g., purchases).
Someone might naively apply those methods to a unified graph of all behaviors, but this approach loses the semantics of each behavior and overlooks important patterns, such as cascading sequences, resulting in limited ranking scores.

While it is not a graph ranking method, Li et al.~\cite{LiLCYLLD24} have recently proposed BPMR, a pattern-mining approach that enumerates diverse patterns of behavioral paths and estimates the probabilities of items being purchased in a Bayesian manner.
Although BPMR outperforms GNN-based representation learning methods, its performance is limited because it considers only a few steps of paths, ignoring the overall graph structure, and requires heavy computational costs for enumeration.

%% file: 030preliminary.tex
\section{Preliminaries}
\label{sec:preliminaires}

We introduce preliminaries on basic notations, and the problem definition addressed in this work. 

\subsection{Notations}
We describe the basic notations frequently used in this paper, with the related symbols summarized in Table~\ref{tab:symbols}.

\smallsection{Vector and matrix}
We use lowercase bold letters for vectors (e.g., $\vect{a}$) and uppercase bold letters for matrices (e.g., $\vect{A}$).
The $i$-th entry of vector $\vect{a}$ is denoted as $\vect{a}(i)$.
The entry in the $i$-th row and $j$-th column of matrix $\mat{A}$ is denoted as $\mat{A}(i, j)$.
The $i$-th row vector of $\mat{A}$ is denoted as $\mat{A}(i)$.

\smallsection{User-item interactions}
Let $\users$ and $\items$ denote the sets of users and items, where $|\users|$ and $|\items|$ are the numbers of users and items, respectively.
Suppose $\B = \{\texttt{view}, \texttt{cart}, \cdots, \texttt{buy}\}$ is the set of behaviors, and let $b_{t}$ denote the target behavior (e.g., $\texttt{buy}$).
If user $u$ has interacted with item $i$ on behavior $b \in \B$, a pair $(u, i)$ is included in $\Eb$, the set of user-item interactions on $b$.

\smallsection{User-item bipartite graphs}
A user-item bipartite graph on behavior $b$ is denoted by $\Gb = (\Ub, \Ib, \Eb)$, where $\Ub$ and $\Ib$ are the copies of $\users$ and $\items$, respectively.
Its bi-adjacency matrix is denoted by $\Ab \in \mathbb{R}^{|\users| \times |\items|}$, where $\Ab(u, i)$ is $1$ if the edge between $u$ and $i$ is in $\Eb$; otherwise, $0$.
Let $\NIb{u}$ be the set of neighboring nodes of $u \in \Ub$, where the neighbors belong to $\Ib$, and its size $d_b(u) = |\NIb{u}|$ is the degree of $u$ in $\Gb$. 
Similarly, $\NUb{i}$ is the set of neighbors of $i \in \Ib$, and $d_b(i) = |\NUb{i}|$ is the degree of $i$ in $\Gb$.

 \begin{table}
    \centering
    \small
    \caption{\label{tab:symbols}Frequently-used symbols.}
    \begin{tabular}{c|l}
        \toprule
        \textbf{Symbol} & \textbf{Description} \\
        \midrule
        $\users$ and $\items$ & sets of users or items, resp.\\
        $\B$ & set of behaviors \\         
        $b_t$ & target behavior (e.g., \texttt{buy}) \\ 
        $q$ & querying user \\ 
        $\C$ & cascading sequence of behaviors, i.e.,  $\C= (b_1 \rightarrow b_2 \rightarrow \cdots 
        \rightarrow b_t)$ \\
        $\Ub$ & set of users on behavior $b$ (i.e., copy of $\users$) \\
        $\Ib$ & set of items on behavior $b$ (i.e., copy of $\items$) \\
        $\Eb$ & set of user-item interactions on behavior $b$\\
        $\Gb$ & user-item graph on behavior $b$, i.e.,  $\Gb = (\Ub, \Ib, \Eb)$ \\
        $\GC$ & cascading behavior graph\\
        $\mat{A}_{b} \in \mathbb{R}^{|\Ub| \times |\Ib|}$ & bi-adjacency matrix of $\Gb$\\
        $\NIb{u}$ and $\NUb{i}$ & sets of neighbors of $u \in \Ub$ and $i \in \Ib$, resp. \\
        $\Dub$ and $\Dib$ & diagonal degree matrices of users and items on $\Gb$, resp.\\
        $\Abnorm$ & symmetrically normalized matrix of $\Ab$, i.e., $\Abnorm = \Dub^{-1/2} \Ab \Dib^{-1/2}$\\ 
        $\rub$ and $\rib$ & ranking score vectors of users and items on $\Gb$, resp. \\
        $\qub$ and $\qib$ & query vectors of users and items on $\Gb$, resp. \\
        $\alpha$ and $\beta$ & strengths of query fitting and cascading alignment \\
        $\gamma$ & strength of smoothness, i.e., $\gamma = 1 - \alpha - \beta$\\
        \bottomrule
    \end{tabular}
\end{table}   

\subsection{Problem Definition}
We describe the formal definition of the problem on multi-behavior recommendation as follows:
\begin{problem}[Personalized Ranking for Multi-behavior Recommendation]
{\color{white}T}
\label{prob:mbr}
\begin{itemize}[leftmargin=9mm,noitemsep]
    \item {
        \textbf{Input}: the sets $\users$ and $\items$ of users and items, the set of multi-behavior interactions, i.e., $\set{E} = \{\set{E}_{b} \: | \: b \in \B \}$, and a querying user $q \in \users$,
    }
    \item {
        \textbf{Output}: a ranking score $\vect{r}_{u}(i)$, meaning the likelihood of the user $q$ performing the target behavior $t$ (e.g., \texttt{buy}) for item $i$.
    }
\end{itemize}
\end{problem}
For each querying user, the recommendation list is generated by sorting items in descending order based on their ranking scores.

\subsection{Personalized Ranking on Bipartite Graphs}
\label{sec:preliminaires:birank}
Given a user-item bipartite graph $\graph=(\users, \items, \edges)$ with single-type interactions, traditional ranking models~\cite{PanPYFD04,BrinB98,DengDLK09kqpg,HeHGKW17} aim to calculate personalized ranking scores regarding a querying user $q$ by analyzing the relationships in $\edges$ between users and items.
For this purpose, most of them follow \textit{smoothness assumption} and \textit{query fitting}, where the former assumes that a node should be ranked high if it is linked to higher-ranked nodes, while the latter prioritizes (or fits) the prior belief of $q$ during the ranking process.

Let $\ru \in \mathbb{R}^{|\users|}$ and $\ri \in \mathbb{R}^{|\items|}$ denote the ranking score vectors for the nodes in $\users$ and $\items$, respectively, with respect to $q$.
For each node $u \in \users$ and $i \in \items$, their respective scores, $\ru(u)$ and $\ri(i)$, are represented as follows:
\begin{align}
    \label{eq:birank}
    \begin{split}
    \ru(u) &= (1-\lambda_{\users}) \cdot
    \Big(
         \sum_{i \in \NIb{u}} \wiu \cdot \ri(i) \
    \Big) + \lambda_{\users}\cdot\qu(u), \\
    \ri(i) &= \underbrace{(1-\lambda_{\items}) \cdot \Big(
    \sum_{u \in \NUb{i}}{\wui \cdot \ru(u)} 
    \Big)}_{\text{Smoothing ranking scores}}
    + \underbrace{\lambda_{\items} \cdot \qi(i), 
        {\color{white}\Biggl(_{B_{B_{B}}}\!\!\!\!\!\!\!\!\!\!\!\!\!\!}
    }_{\text{Fitting queries}} 
    \end{split}
\end{align}
where $\qu(u)$ and $\qi(i)$ represent the prior beliefs at nodes $u$ and $i$, respectively, given by the querying user $q$.
The first term smooths (or aggregates) the scores of the target node’s neighbors with normalized edge weights $\wiu$ and $\wui$, while the second term injects the prior beliefs in $\qu(u)$ and $\qi(i)$ to the ranking scores. 
Their contributions are adjusted by $\lambda_{\users} \in [0,1]$ and $\lambda_{\items} \in [0,1]$, called personalized parameters for users and items, respectively. 
The ranking models iteratively refines these scores from their initial values until convergence.
Note that different designs of 1) the normalized edge weights $\{\wui, \wiu\}$, 2) the personalized parameters $\{\lambda_{\users}, \lambda_{\items}\}$, or 3) the query vectors $\{\qu, \qi\}$ lead to different ranking models (refer to~\cite{HeHGKW17} for detailed configuration information).

%% file: 040method.tex
\section{Proposed Method}
\label{sec:proposed}
In this section, we propose \method, a novel graph ranking method for effective multi-behavior recommendation.

\subsection{Overview}
\label{sec:proposed:overview}
We summarize the technical challenges addressed in this work for accurate multi-behavior recommendation as follows:

\begin{enumerate}[leftmargin=8mm,noitemsep]
    \item[C1.]{
        \textbf{Achieving accurate recommendations.}
        The representation learning methods for multi-behavior recommendation have limited expressive power, leading to suboptimal accuracy. 
        How can we enhance recommendation accuracy in multi-behavior settings, especially compared to these representation learning methods?
    }
    \item[C2.]{
        \textbf{Leveraging multi-behavior characteristics.}
        Graph ranking can be a promising alternative to representation learning, but the traditional ranking methods are limited in multi-behavior settings because they are designed for single-behavior graphs.
        Which aspects of multi-behavior interactions can be leveraged for effective graph ranking?
    }
    \item[C3.]{
        \textbf{Ensuring stability and scalability.}
        Ensuring stable ranking results and scalable computation is crucial when designing a personalized ranking method on graphs.
        How can we develop a graph ranking algorithm tailored for multi-behavior recommendation that guarantees both stability and scalability?
    }
\end{enumerate}

We propose the following ideas that addresses the aforementioned technical challenges:

\begin{enumerate}[leftmargin=8mm,noitemsep]
    \item[A1.]{
        \textbf{Personalized graph ranking.}
        We adopt a graph ranking approach for accurate multi-behavior recommendation. 
        Our ranking design adheres to collaborative filtering by smoothing ranking scores and fitting queries, while mitigating the impact of high-degree nodes on our ranking scores.
    }
    \item[A2.]{
        \textbf{Cascading behavior graph and cascading alignment.}
        We leverage a cascading sequence of user behaviors, constructing a cascading behavior graph that sequentially connects behavior graphs, and measure rankings on it. 
        This smoothly aligns the ranking scores of the current behavior with those of the previous one, capturing the cascading effect.
    }
    \item[A3.]{
        \textbf{Iterative computation.}
        We iteratively refine the ranking scores of \method for a querying user, guiding them to converge (i.e., reach a stable state) while ensuring scalable computation w.r.t. the number of interactions.
    }
\end{enumerate}

\begin{figure}[t]
    \centering
    \includegraphics[width=1.0\linewidth]{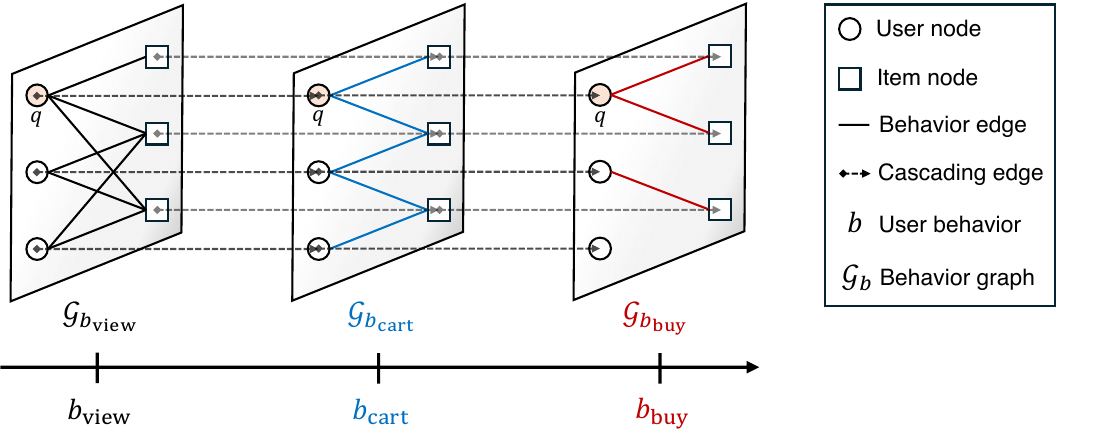}
    \caption{
        \label{fig:overview}
        Example of a cascading behavior graph $\GC=(\mathcal{G}_{b_{\texttt{view}}} \rightarrow \mathcal{G}_{b_{\texttt{cart}}} \rightarrow \mathcal{G}_{b_{\texttt{buy}}})$, given a cascading sequence $\C = (b_{\texttt{view}} \rightarrow b_{\texttt{cart}} \rightarrow b_{\texttt{buy}})$ of behaviors where $b_{\texttt{buy}}$ is the target behavior (i.e., $b_t$), and $q$ is the querying user node.
        Starting from $\mathcal{G}_{b_{\texttt{view}}}$, \method calculates ranking scores w.r.t. $q$ along $\GC$, and returns ranking scores of items at the target behavior $b_t$.
    }
\end{figure}

To model a cascading sequence of behaviors, we first construct a \textit{cascading behavior graph} by connecting behavior graphs in the order of the sequence, as shown in Figure~\ref{fig:overview}.
Then, we measure the ranking scores along the cascading behavior graph while considering the querying user’s past history for each behavior. 
We use the ranking scores on the target (or last) behavior to recommend items to the querying user.

\subsection{Cascading Behavior Graph Construction}
\label{sec:proposed:cbg}

To design a graph ranking model for multi-behavior interactions, one might simply merge the interactions, resulting in a unified multigraph where each edge represents either a unified or distinct type of behavior, and then measure the ranking scores on the graph. 
However, this naive approach may fail to fully capture the semantics of user behaviors. 
Especially, it is worth noting that behaviors often occur in a certain sequence~\cite{LiuXWY00024, yin2024hecgcn, ChengCHLZGP23fqvn, YanCGSLSL24}; for example, a user first views an item, adds it to a cart, and then purchases it.
We aim to inject this information of user behaviors into our ranking scores. 
%
For this purpose, we first define a cascading sequence as follows:
\begin{definition}[Cascading Sequence]
\label{def:cascading_sequence}
Given a set $\B$ of user behaviors, a cascading sequence $\C$ is defined as follows:
\begin{equation*}
    \C \coloneqq (b_1 \rightarrow b_2 \rightarrow \cdots \rightarrow b_{|\B|}), 
\end{equation*}
where each behavior $b_{i} \in \B$ in the sequence is distinct. 
In general, a sequence of user behaviors leads to the target behavior $b_t$ (e.g., buy); thus, the last behavior $b_{|\B|}$ of $\C$ represents $b_t$ in this case.
\hfill\qedsymbol
\end{definition}

We consider a natural cascading sequence $\C$ of behaviors, i.e., $\C = (\texttt{view} \rightarrow \texttt{cart} \rightarrow \texttt{buy})$ for the aforementioned case. 
From the cascading sequence, we construct the cascading behavior graph $\GC$ as follows:
\begin{definition}[Cascading Behavior Graph]
Given a cascading sequence $C = (b_1 \rightarrow b_2 \rightarrow \cdots \rightarrow b_{t-1} \rightarrow b_t)$ and a set $\{\Gb \: | \: b \in \B \}$ of behavior graphs, the cascading behavior graph $\GC$ is defined as follows:
\begin{align}
    \label{eq:cbg}
    \GC \coloneqq (\Gind{b_1} \rightarrow \Gind{b_2} \rightarrow \cdots \rightarrow
    \Gind{b_{t-1}} \rightarrow \Gind{b_t}), 
\end{align}
where each $\Gind{b'} \rightarrow \Gind{b}$ indicates a node-wise connection that forwardly links each node in $\Gind{b'}$ of the previous behavior $b'$ to the corresponding node in  $\Gind{b}$ of the next behavior $b$.
\hfill\qedsymbol
\end{definition}

Figure~\ref{fig:overview} presents an example of the cascading behavior graph $\GC$. 
By doing so, we ensure that information from a previous behavior influences that of the subsequent behavior in the sequence $\C$, thereby modeling temporal dynamics similar to those observed in dynamic graphs~\cite{LeeLJ22}.
Note that for this purpose, we do not allow reverse-direction connections such as $\Gind{b'} \leftarrow \Gind{b}$ in $\GC$.

\subsection{\method: Personalized Ranking on a Cascading Behavior Graph}
\label{sec:proposed:cascrank}

In this section, we design our ranking model \method on the cascading behavior graph $\GC$ for multi-behavior recommendation. 
Our main idea for estimating ranking scores along the cascading sequence is to smooth the ranking scores of neighboring nodes while incorporating query-specific information for the current behavior $b$, as inspired by Equation~\eqref{eq:birank}, and to propagate and smooth the ranking scores from the previous behavior $b'$.

\begin{figure}[t]
    \centering
    \includegraphics[width=0.9\linewidth]{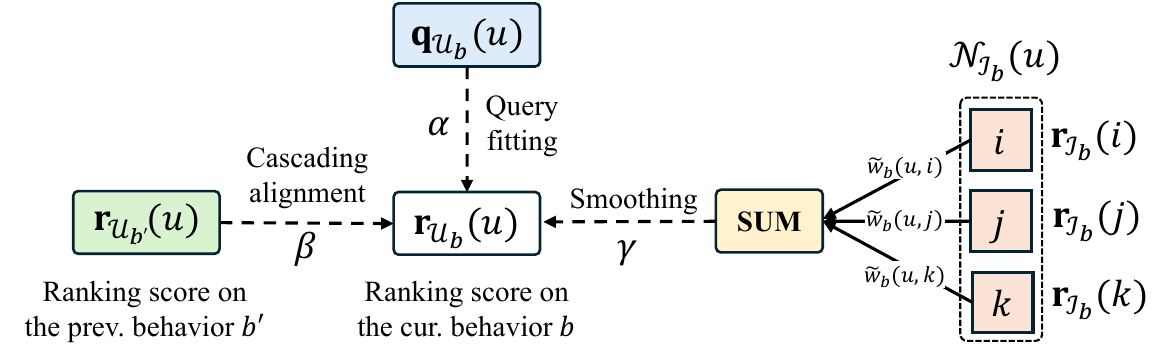}
    \caption{
        \label{fig:formulation}
        Illustration of the computation on the ranking score $\rub(u)$ of user $u$ of Equation~\eqref{eq:cascrank:inst}, with details of the notations provided in Section~\ref{sec:proposed:cascrank}. 
    }
\end{figure}Suppose we consider $\Gind{b'} \rightarrow \Gind{b}$ in the cascading behavior graph $\GC$ to obtain ranking scores $\rub(u)$ and $\rib(i)$ of user $u$ and item $i$ for the current behavior $b$ with respect to the querying user $q$. 
These scores are recursively represented as follows:
\begin{align}
    \label{eq:cascrank:inst}
    \begin{split}
    \rub(u) &= (1-\alpha-\beta) \cdot
    \biggl(
         \sum_{i \in \NIb{u}} \wuib \cdot \rib(i) \
    \biggr) + \alpha\cdot\qub(u) + \beta\cdot\rubp(u), \\
    \rib(i) &= \underbrace{(1-\alpha-\beta) \cdot \biggl(
    \sum_{u \in \NUb{i}}{\wiub \cdot \rub(u)} \biggr)}_{_{\substack{\text{Ranking score} \\ \text{smoothing}}}}
    + \underbrace{\alpha \cdot \qib(i)
        {\color{white}\Biggl(_{B_{B_{B}}}\!\!\!\!\!\!\!\!}
    }_{\substack{\text{Query} \\ \text{fitting}}}  + \underbrace{\beta \cdot \ribp(i), {\color{white}\Biggl(_{B_{B_{B}}}\!\!\!}}_{\substack{\text{Cascading} \\ \text{alignment}}}
    \end{split}
\end{align}
where the equation of $\rub(u)$ is illustrated in Figure~\ref{fig:formulation}, and that of $\rib(i)$ is similar to the figure.
The first term in $\rub(u)$ follows the smoothness assumption by aggregating the ranking scores of each neighbor in $\NIb{u}$ of node $u$ with weight $\wuib$, where it is adjusted by $\gamma \coloneqq 1 - \alpha -\beta$, referred to as the strength of the smoothness.
The second term aims to align with $\qub(u)$ and $\qib(i)$, which encapsulate the information associated with nodes $u$ and $i$ with respect to the querying user $q$ for behavior $b$, where $\alpha$ controls the strength of the query fitting.
The third term, called \textit{cascading alignment}, smoothly aligns $\rub(u)$ with $\rubp(u)$ from the previous behavior $b’$, where $\beta$ controls the strength of this alignment.
The hyperparameters $\alpha$ and $\beta$ are within the range $(0, 1)$, satisfying $0 < \alpha + \beta < 1$.
%
Note that the equation of $\rib(i)$ is interpreted similarly to that of $\rub(u)$, and we set $\alpha = \lambda_{\mathcal{U}} = \lambda_{\mathcal{I}}$, unlike in Equation~\eqref{eq:birank}, to reduce the number of hyperparameters.

\smallsection{Normalization weights}
As discussed in Section~\ref{sec:preliminaires:birank}, the design choice for the normalization edge weights $\wuib$ and $\wiub$ is one of the crucial factors in designing ranking models.
In this work, we adopt a symmetric normalization that sets these weights when there is an interaction between $u$ and $i$ (i.e., $\Ab(u,i)=1$ and $\Ab^{\top}(i,u)=1$), as follows:
\begin{equation*}
    \wuib = \frac{\Ab(u,i)}{\sqrt{d_b(u)}\sqrt{d_b(i)}} \quad \text{and} \quad \wiub = \frac{\Ab^{\top}(i,u)}{\sqrt{d_b(i)}\sqrt{d_b(u)}}, 
\end{equation*}
where $\Ab$ is the bi-adjacency matrix of $\Gb$, and $d_b(u)$ and $d_b(i)$ are the degrees of user $u$ and item $i$, respectively.
If $\Ab(u, i)$ is $0$, then $\wuib = \wiub$ is set to $0$.
The main reason we choose the symmetric normalization is its effectiveness in eliminating the impact of high-degree nodes when aggregating the ranking scores, especially compared to a stochastic normalization, such as $\wuib = \Ab(u,i)/d_b(i)$.
By doing so, this prevents the estimated ranking scores from being biased toward high-degree nodes (or items) and avoids penalizing low-degree items in a long-tail distribution. For this reason, it has been widely used in graph ranking for recommendation in previous studies~\cite{HeHDWLZW20, HeHGKW17}.
We empirically verify the effectiveness of the symmetric normalization over the stochastic normalization, as shown in Table~\ref{tab:ablation:normalization}.

\smallsection{Query information}
We use the information of $q$ and its interaction history for each behavior $b$ to set $\qub$ and $\qib$ for multi-behavior recommendation. 
Specifically, we set $\qub$ and $\qib$ as follows:
\begin{equation*}
    \qub(u) = \begin{cases}
        1, & \text{if } u = q, \\
        0, & \text{otherwise, }
    \end{cases}
    \quad \text{ and } \quad
    \qib(i) = \begin{cases}
        |\NIb{q}|^{-1}, & \text{if } i \in \NIb{q}, \\
        0, & \text{otherwise, }
    \end{cases}
\end{equation*}
where $\NIb{q}$ denotes the set of items that $q$ has interacted with under behavior $b$.
In other words, $\qub$ aims to enhance personalization for $q$, while $\qib$ incorporates the user's interaction history into the ranking scores.

\smallsection{Vectorization} 
We vectorize the entry-wise form of Equation~\eqref{eq:cascrank:inst} using the following matrix-vector multiplications:
\begin{align}
    \label{eq:cascrank:vect}
    \begin{split}
    \rub &= \gamma \cdot \Abnorm \cdot \rib \
    + \alpha\cdot\qub + \beta\cdot\rubp, \\
    \rib &= \gamma \cdot
        \Abnorm^{\top} \cdot \rub \
     + \alpha\cdot\qib + \beta\cdot\ribp,
    \end{split}
\end{align}
where $\gamma = 1 - \alpha - \beta$, and
$\Abnorm = \Dub^{-1/2} \Ab \Dib^{-1/2}$ is the symmetrically normalized bi-adjacency matrix. 
$\Dub$ and $\Dib$ are the diagonal degree matrices of $\Ub$ and $\Ib$, respectively, defined as follows:
\begin{equation*}
    \Dub = \texttt{diag}\Big(d_{b}(u) \: | \: \forall u \in \Ub \Big), \quad \text{and} \quad 
    \Dib = \texttt{diag}\Big(d_{b}(i) \: | \: \forall i \in \Ib \Big)
\end{equation*}
where $\texttt{diag}(\cdot)$ returns a diagonal matrix with the input values on its diagonal.
Then, Equation~\eqref{eq:cascrank:vect} can be further compactly represented as follows:
\begin{align}
    \label{eq:cascrank:mat}
    \underset{\rb}
    {\begin{bmatrix}
        \rub\\\rib
    \end{bmatrix}}=
    \gamma \cdot 
    \underset{\tilde{\mathbf{\mathcal{A}}}_b}{\begin{bmatrix}
        \mat{0}_{\Ub} & \Abnorm \\
        \Abnorm^{\top} & \mat{0}_{\Ib}
    \end{bmatrix}}
    \underset{\rb}
    {\begin{bmatrix}
        \rub\\\rib
    \end{bmatrix}}+
    \alpha \cdot
    \underset{\qb}
    {\begin{bmatrix}
        \qub\\\qib
    \end{bmatrix}}+
    \beta\cdot
    \underset{\vect{r}_{b'}}
    {\begin{bmatrix}
        \rubp \\ \ribp
    \end{bmatrix}},
\end{align}
where $\mat{0}_{\Ub} \in \mathbb{R}^{|\Ub| \times |\Ub|}$ and $\mat{0}_{\Ib} \in \mathbb{R}^{|\Ib| \times |\Ib|}$ are zero matrices, and $\rb$ contains the ranking scores for all users and items, whose closed-form solution is given as follows:
\begin{align}
    \label{eq:cascrank:closed}
    \begin{split}
    \big(\mat{I} - \gamma \cdot \tilde{\mathbf{\mathcal{A}}}_b \big) \cdot \rb &= \alpha \cdot \qb + \beta \cdot \vect{r}_{b'}, \\
    \mat{r}_{b} &= \mat{L}_b^{-1} \cdot
    \big(
        \alpha \cdot \qb + \beta \cdot \vect{r}_{b'}
    \big),
    \end{split}
\end{align}
where 
$\mat{I}$ is an identity matrix, and
$\mat{L}_b = \mat{I} - \gamma \cdot \tilde{\mathbf{\mathcal{A}}}_b$ is invertible for $0 < \gamma < 1$ because it is positive definite
\footnote{
The symmetrically normalized adjacency matrix $\mathcal{A}_b$ is symmetric with eigenvalues $\lambda$ in $[-1, 1]$, and the eigenvalues of $\mathbf{L}_b$, represented as $1 - \gamma \cdot \lambda \in [1 - \gamma, 1 + \gamma]$, are all positive for $0 < \gamma < 1$, indicating that $\mathbf{L}_b$ is positive definite.
}.
Given $\GC = (\Gind{b_1} \rightarrow \Gind{b_2} \rightarrow \cdots \rightarrow \Gind{b_t})$, we first compute $\vect{r}_{b_1}$, initializing the previous score vector $\vect{r}_{b'}$ to $\vect{q}_{b_1}$. 
We then proceed along the chain, repeatedly computing $\rb$ for each behavior, until finally obtaining $\vect{r}_{b_t}$ for the target behavior $b_t$.

\subsection{Iterative Algorithm for \method}
\label{sec:proposed:alg}
Although the ranking scores can be directly computed using Equation~\eqref{eq:cascrank:closed}, the matrix inversion operation of $\mat{L}_{b}^{-1}$ is intractable for large-scale graphs due to its cubic time complexity.
Instead, we employ an iterative approach, such as power iteration, similar to other graph ranking methods~\cite{HeHDWLZW20, HeHGKW17, BrinB98}. 
The main idea is to repeatedly update Equation~\eqref{eq:cascrank:vect} starting from initial values until convergence, which is encapsulated in Algorithm~\ref{alg:method}.
It first initializes $\rubp$ and $\ribp$ as the query vectors (line~\ref{alg:method:init}), as there is no previous behavior for the first behavior in $\GC$.
Then, it computes $\rub$ and $\rib$ for each behavior $b$ along the chain in $\GC$ (lines~\ref{alg:method:outer:for}-\ref{alg:method:outer:forend}).
Suppose we consider $\Gind{b}$ of the link $\Gind{b'} \rightarrow \Gind{b}$ in the chain of $\GC$ (line~\ref{alg:method:outer:for}). 
Then, it symmetrically normalizes $\Ab$ (line~\ref{alg:method:norm}), and then performs the power iteration to obtain $\rub$ and $\rib$ (lines~\ref{alg:method:inner:for}-\ref{alg:method:inner:forend}).
It repeats Equation~\eqref{eq:cascrank:vect} at most $K$ times until the convergence criteria is satisfied (line~\ref{alg:method:convergence}).
After updating $\rubp$ and $\ribp$ for the next iteration (line~\ref{alg:method:update:next}), it continues the power iterations until reaching the last (or target) behavior.
Since we only need the ranking scores of items for the target behavior $b_t$ for Problem~\ref{prob:mbr}, it returns $\vect{r}_{\items_{b_t}}$ at the end (line~\ref{alg:method:return}).
Note that the power iteration converges when $\alpha$, $\beta$, and $\gamma$ are in $(0, 1)$, as guaranteed by Theorem~\ref{theorem:convergence}.

\begin{algorithm}[t]
    \caption{Iterative algorithm of \method}
    \label{alg:method}
    \begin{algorithmic}[1]
        \small
        \Require 
        \item[] Cascading behavior graph: $\GC$
        \item[] Querying user: $q$
        \item[] Strength of query fitting: $\alpha$
        \item[] Strength of cascading alignment: $\beta$
        \item[] Number of iterations: $K$
        \item[] Convergence threshold: $\epsilon$
        \Ensure 
        \item[] Ranking score vector $\vect{r}_{\items_{b_t}}$ for items on the target behavior $b_t$ w.r.t. $q$
        \vspace{1mm} 
        \State Let $b'$ and $b$ denote the previous and current behaviors, respectively
        \State Set $\rubp \leftarrow \vect{q}_{\users_{b_1}}$ and $\ribp \leftarrow \vect{q}_{\items_{b_1}}$ \label{alg:method:init}{\Comment{\textbf{\scriptsize
        Initialization for the cascading alignment}}}
        
        \For{\textbf{each} $\Gind{b} \in \mathcal{\GC}$} \label{alg:method:outer:for}
            \State Symmetrically normalize $\Ab$, i.e., $\Abnorm \leftarrow \Dub^{-1/2} \Ab \Dib^{-1/2}$  \label{alg:method:norm}
            \State Set $\rub' \leftarrow \qub$ and $\rib' \leftarrow \qib$ {\Comment{\textbf{\scriptsize Initialization for the power iteration}}} 
            
            \For {$k \leftarrow 1 $ to $K$} \label{alg:method:inner:for} {\Comment{\textbf{\scriptsize Power iteration}}} 
            \label{alg:method:poweriteration}
                \State $\rub \leftarrow \gamma \cdot \Abnorm \cdot \rib' + \alpha\cdot\qub + \beta\cdot\rubp$ \label{alg:method:userscore}
                {\Comment{\textbf{\scriptsize $\gamma = 1-\alpha-\beta$}}} 
                \State $\rib \leftarrow \gamma \cdot\Abnorm^{\top} \cdot \rub'  + \alpha\cdot\qib + \beta\cdot\ribp$
                \label{alg:method:itemscore}
                \If{$\lVert \rub - \rub' \rVert_{1}+ \lVert \rib - \rib' \rVert_{1} \leq \epsilon$} \label{alg:method:convergence}
                    \State \textbf{break}
                \EndIf
                \State Set $\rub' \leftarrow \rub$ and $\rib' \leftarrow \rib$ {\Comment{\textbf{\scriptsize Update for the power iteration}}} 
            \EndFor \label{alg:method:inner:forend}
            \State Set $\rubp \leftarrow \rub$ and $\ribp \leftarrow \rib$ \label{alg:method:update:next}
            {\Comment{\textbf{\scriptsize Update for the cascading alignment}}} 
        \EndFor \label{alg:method:outer:forend}
    \State \textbf{return} $\vect{r}_{\items_{b_t}} \leftarrow \ribp$ \label{alg:method:return}
    \end{algorithmic}
\end{algorithm}

\subsection{Interpretation from the Perspective of Optimization}
\label{sec:proposed:reg}
We interpret the ranking design of \method in the form of an optimization problem, inspired by manifold optimization~\cite{ZhouZBLWS03,AgarwalA06, HeHGKW17}. 
Based on the notations in Equation~\eqref{eq:cascrank:mat}, our ranking model minimizes the following objective function:
\begin{align}
    \label{eq:reg:obj}
    \mathcal{J}(\rb; \mathcal{A}_{b}, \qb, \rbp) = \underbrace{\rb^{\top}(\mat{I} - \mathcal{A}_{b})\rb}_{_{\substack{\text{Ranking score} \\ \text{smoothing}}}} + \underbrace{\theta\lVert\rb - \qb \rVert_{2}^{2}}_{\substack{\text{Query} \\ \text{fitting}}} + \underbrace{\omega\lVert\rb - \rbp \rVert_{2}^{2}}_{\substack{\text{Cascading} \\ \text{alignment}}}, 
\end{align}
where the first term, $\rb^{\top}(\mat{I} - \mathcal{A}_{b})\rb = \sum_{i,j}\mathcal{A}_b(i, j)\cdot\lVert\rb(i) - \rb(j)\rVert_{2}^{2}$, is referred to as \textit{graph Laplacian smoothing}, which aims to measure the difference in the ranking scores between two connected nodes.
$\theta$ and $\omega$ controls the strength of the regularizations of query fitting and cascading alignment, respectively.

The solution $\rb$ obtained by \method minimizes the optimization problem, effectively smoothing the scores of neighboring nodes while also fitting the query information in $\qb$ and the previous scores in $\rbp$.
To prove the claim, we rewrite the objective function as follows:
\begin{align*}
    \begin{split}
    \mathcal{J}(\rb; \mathcal{A}_{b}, \qb, \rbp)  &= 
    \begin{bmatrix}
        \rub^{\top} & \rib^{\top}
    \end{bmatrix}
    \begin{bmatrix}
        \mat{I} & -\tAb \\
        -\tAb^{\top} & \mat{I} 
    \end{bmatrix}
    \begin{bmatrix}
        \rub \\ \rib
    \end{bmatrix} \\
    &+ \theta \cdot \left( \lVert \rub - \qub \rVert_{2}^{2} + \lVert \rib - \qib \rVert_{2}^{2} \right) \\
    &+ \omega \cdot \left(\lVert \rub - \rubp \rVert_{2}^{2} + \lVert \rib - \ribp \rVert_{2}^{2} \right). 
    \end{split}
\end{align*}

Note that the optimization problem of graph Laplacian smoothing with strictly convex regularization terms (e.g., $\lVert \rub - \qub \rVert_{2}^{2}$) is strictly convex~\cite{HeHGKW17}, and thus it has a unique global minimum.
The first-order derivatives with respect to $\rub$ and $\rib$ are represented as follows: 
\begin{align*}
    \begin{split}
    \frac{\partial\mathcal{J}}{\partial\rub} = 
    2\rub - 2\tAb\rib + 2\theta (\rub - \qub) + 2\omega(\rib - \ribp), \\
    \frac{\partial\mathcal{J}}{\partial\rib} = 
    2\rib - 2\tAb^{\top}\rub + 2\theta (\rib - \qib) + 2\omega(\rib - \ribp).
    \end{split}
\end{align*}

By setting each derivative to $0$, it provides the following solutions:
\begin{align*}
    \begin{split}
        \rub = \frac{1}{1+\theta+\omega}\tAb\rib + 
        \frac{\theta}{1+\theta+\omega}\qub + \frac{\omega}{1+\theta+\omega}\rubp, \\
        \rib = \frac{1}{1+\theta+\omega}\tAb^{\top}\rub + 
        \frac{\theta}{1+\theta+\omega}\qib + \frac{\omega}{1+\theta+\omega}\ribp,
    \end{split}
\end{align*}
which are the same as Equation~\eqref{eq:cascrank:vect} if $\theta=\frac{\alpha}{1-\alpha-\beta}$ and $\omega = \frac{\beta}{1-\alpha-\beta}$.

\subsection{Theoretical Analysis}
\label{sec:proposed:analysis}
In this section, we theoretically analyze \method and its iterative algorithm with the following questions:
\begin{enumerate}[leftmargin=9mm,noitemsep]
    \item[T1.] {
        \textbf{Cascading effect.}
        Does our ranking design adhere to the cascading effect, where a behavior occurring later in the sequence has a greater influence on the target behavior?
    }
    \item[T2.] {
        \textbf{Convergence guarantee.}
        Is the convergence of our iterative algorithm guaranteed? 
    }
    \item[T3.] { 
        \textbf{Computational complexity.} Does our iterative algorithm scale linearly with the number of interactions?
    }
\end{enumerate}

\subsubsection{Analysis on Cascading Effect (T1)}
\label{sec:proposed:analysis:casc}
In this section, we investigate the cascading effect of \method by analyzing the ranking score vector $\vect{r}_{b_t}$ on the target behavior $b_t$ based on Equation~\eqref{eq:cascrank:closed} in the following theorem.

\begin{theorem}[Cascading Effect] 
\label{theorem:decay}
The ranking score vector $\vect{r}_{b_t}$ on the target behavior $b_t$ is represented as follows:
\begin{align}
    \label{eq:casc_effect}
    \vect{r}_{b_t} = \left(\sum_{i=0}^{t-1}\beta^{i} \cdot \mat{\mathcal{L}}_{b_t \leftsquigarrow b_{t-i}} \cdot \hat{\vect{q}}_{b_{t-i}} \right) + \beta^{t}\cdot\mat{\mathcal{L}}_{b_t \leftsquigarrow b_1}\cdot\vect{r}_{b_0},
\end{align}
where $\hat{\vect{q}}_{b_i} = \alpha \cdot \vect{q}_{b_i}$, $\mat{\mathcal{L}}_{b_i} = \mat{L}_{b_i}^{-1} = (\mat{I} - \gamma \cdot \tilde{\mathbf{\mathcal{A}}}_{b_i})^{-1}$, $\mat{\mathcal{L}}_{b_j\leftsquigarrow b_i} = \mat{\mathcal{L}}_{b_j} \cdot \mat{\mathcal{L}}_{b_{j-1}}\cdots\mat{\mathcal{L}}_{b_i}$, 
$\Linvprod{i}{i} = \mat{\mathcal{L}}_{b_i}$,
and $\vect{r}_{b_0}$ is an initial vector of ranking scores (i.e., $\vect{q}_{b_1}$).
\end{theorem}
\begin{proof}
    We prove the claim by mathematical induction.
    As a base case when $b_t = b_1$,  $\vect{r}_{b_1}$ is represented as follows:
    \begin{align*}
        \vect{r}_{b_1} = \Linvprod{1}{b_1}\cdot\qindhat{1} + \beta\cdot\Linvprod{1}{1}\cdot\vect{r}_{b_0} \Leftrightarrow \vect{r}_{b_1} = \L^{-1}_{b_1}(\alpha \cdot \vect{q}_{b_1}+\beta \cdot \vect{r}_{b_0}).
    \end{align*}
    Thus, it obviously holds for Equation~\ref{eq:cascrank:closed} when applied to $\vect{r}_{b_1}$.
    Suppose the claim holds at behavior $b_t$. 
    Based on Equation~\eqref{eq:cascrank:closed}, $\vect{r}_{b_{t+1}}$ is represented as follows:
    \begin{align*}
        \mat{r}_{b_{t+1}} &= \mat{L}_{b_{t+1}}^{-1} \cdot
    \big(
        \alpha \cdot \vect{q}_{b_{t+1}} + \beta \cdot \vect{r}_{b_{t}}
    \big) \Leftrightarrow
    \vect{r}_{b_{t+1}} = \Linv_{b_{t+1}}\cdot\qindhat{t+1} + \beta \cdot \Linv_{b_{t+1}}\cdot\vect{r}_{b_t}
    \end{align*}
    By replacing $\vect{r}_{b_t}$ of Equation~\eqref{eq:casc_effect} into the above equation, it is represented as follows: 
    \begin{align*}
        \vect{r}_{b_{t+1}} &= \Linv_{b_{t+1}}\cdot\qindhat{t+1} + \beta\cdot \Linv_{b_{t+1}}\cdot\Biggl(\sum_{i=0}^{t-1}\beta^{i}\cdot\Linvprod{t}{t-i}\cdot\qindhat{t-i} + \beta^{t}\cdot\Linvprod{t}{1}\cdot\vect{r}_{b_0} \Biggl),\\
         &= \sum_{i=0}^{t}\beta^{i}\cdot\Linvprod{t+1}{t+1-i}\cdot\qindhat{t+1-i} + \beta^{t+1}\cdot\Linvprod{t+1}{1}\cdot\vect{r}_{b_0}.
    \end{align*}
    If we replace $t+1$ with $k$, then it is represented as follows:
    \begin{align*}
        \vect{r}_{b_k} = \sum_{i=0}^{k-1} \beta^{i}\cdot\Linvprod{k}{k-i}\cdot\qindhat{k-i} + \beta^{k}\cdot\Linvprod{k}{1}\cdot\vect{r}_{b_0}. 
    \end{align*}
    This equation has exactly the same form as Equation~\eqref{eq:casc_effect}. 
    Thus, it holds for $b_{t+1}$, and therefore, the claim is true by mathematical induction.
\end{proof}

From Theorem~\ref{theorem:decay}, the ranking score vector $\vect{r}_{b_t}$ on the target behavior can be represented as follows:
\begin{align*}
    \vect{r}_{b_t} \propto \beta^{0}\cdot\Linvprod{t}{t}\cdot\qindhat{t} + \beta^{1}\cdot\Linvprod{t}{t-1}\cdot\qindhat{t-1} + \cdots + \beta^{t-1}\cdot\Linvprod{t}{1}\cdot\qindhat{1}, 
\end{align*}
where the term $\Linvprod{t}{t-i}$ is interpreted as the diffused propagation of the information in $\qindhat{t-i}$ from $\mathcal{G}_{b_{t-i}}$ to $\mathcal{G}_{b_{t}}$ along the chain, and the result decays as $i$ increases, with a factor of $0\leq \beta \leq 1$.
This aligns with the cascading effect, indicating that behaviors with stronger preference signals (e.g., purchases) are emphasized more than those with weaker preference signals (e.g., views).
A smaller value of $\beta$ further attenuates the impact of earlier behaviors in the chain on the ranking scores. 

\subsubsection{Analysis on Convergence (T2)}
\label{sec:proposed:analysis:convergence}
Inspired by~\cite{HeHGKW17}, we theoretically analyze the convergence of the iterative algorithm in Algorithm~\ref{alg:method} for \method in the following theorem.

\begin{theorem}[Convergence] 
\label{theorem:convergence}
The power iteration in Algorithm~\ref{alg:method} converges when $\gamma$ is in $(0, 1)$.
\end{theorem}
\begin{proof}
    Suppose we consider $\Gind{b}$ of the link $\Gind{b'} \rightarrow \Gind{b}$ in the chain of $\GC$.
    In Equation~\eqref{eq:cascrank:vect}, we substitute $\rub$ into the equation for $\rib$, which is represented as follows:
    \begin{align*}
        \rib &\leftarrow 
        \underbrace{\gamma^2 \cdot \Abnorm^{\top}\cdot\Abnorm \cdot\rib}_{\text{Variant term}}
        + \underbrace{\alpha\gamma \cdot \Abnorm^{\top}\cdot\qub + \beta\gamma\cdot \Abnorm^{\top}\cdot\rubp + \alpha \cdot \qib + \beta \cdot \ribp}_{\text{Invariant term}}, \\
        \rib &\leftarrow \mat{S}\cdot\rib + \vect{x},
    \end{align*}
where $\gamma = 1-\alpha-\beta$, and $\rubp$ and $\ribp$ are fixed at this step.
Let $\mat{S}$ represent $\gamma^{2} \cdot \Abnorm^{\top} \cdot \Abnorm$ in the variant term, and let $\vect{x}$ denote the invariant term.
Suppose $\rib^{(k)}$ represents the result after $k$ iterations of the above equation, starting from $\rib^{(0)}$. It can then be represented as follows:
\begin{align*}
    \rib^{(k)} = \mat{S}^{k}\cdot\rib^{(0)} + \sum_{t=0}^{k-1} \mat{S}^{t} \cdot \vect{x}.
\end{align*}
Assuming the eigenvalues of $\mat{S}$ lie within the range $(-1, 1)$, the infinite iterations for $\rib^{(k)}$ converge as follows:
\begin{align*}
\lim_{k\rightarrow\infty}\mat{S}^{k}\cdot\rib^{(0)} = 0, \quad \text{and} \quad \lim_{k\rightarrow\infty}\sum_{t=0}^{k-1}\mat{S}^{k} = (\mat{I}-\mat{S})^{-1},
\end{align*}
which is guaranteed by the Neumann series (i.e., the geometric series for matrices).
To analyze the range of the eigenvalues of $\mat{S}$, we represent it as follows:
\begin{align*}
    \mat{S} &= \gamma^{2}\cdot \Abnorm^{\top}\cdot \Abnorm, \\
    &= \gamma^{2}\cdot\Big(\mat{D}_{\users_b}^{-1/2}\Ab\mat{D}_{\items_b}^{-1/2}\Big)^{\top}\cdot\Big(\mat{D}_{\users_b}^{-1/2}\Ab\mat{D}_{\items_b}^{-1/2} \Big), \\
    &= \gamma^2\cdot \Big(\mat{D}_{\items_b}^{-1/2}\Ab^{\top}\mat{D}_{\users_b}^{-1}\Ab\mat{D}_{\items_b}^{-1/2} \Big), \\
    &= \mat{D}_{\items_b}^{-1/2}\mat{S}_{v}\mat{D}_{\items_b}^{1/2},
\end{align*}
where $\mat{S}_{v} = \gamma^{2}(\Ab^{\top}\mat{D}_{\users_b}^{-1}\cdot\Ab\mat{D}_{\items_b}^{-1})$. 
Note that the largest absolute eigenvalue of  $\Ab^{\top}\mat{D}_{\users_b}^{-1}\cdot\Ab\mat{D}_{\items_b}^{-1}$ is $1$ because it is a column-stochastic matrix, and thus, the eigenvalues of $\mat{S}_{v}$ are in the range $[-\gamma^2, \gamma^2]$.
In addition, the eigenvalues of $\mat{S}_{v}$ are identical to those of $\mat{S}$. 
The reason is as follows:
\begin{align*}
    |\mat{S}_{v} - \lambda_{v}\mat{I}| = |\mat{D}_{\items_b}^{1/2}(\mat{S} - \lambda_{v}\mat{I})\mat{D}_{\items_b}^{-1/2}| = 
    |\mat{D}_{\items_b}^{1/2}|\cdot|\mat{S} - \lambda_{v}\mat{I}|\cdot |\mat{D}_{\items_b}^{-1/2}| = |\mat{S} - \lambda_{v}\mat{I}| = 0,
\end{align*}
where $|\cdot|$ indicates the determinant of a given matrix. 
The largest absolute eigenvalue of  $\Ab^{\top}\mat{D}_{\users_b}^{-1}\cdot\Ab\mat{D}_{\items_b}^{-1}$ is $1$ because it is column stochastic.
Therefore, the eigenvalues of $\mat{S}$ lie within the range $[-\gamma^2, \gamma^2]$, and the condition $\gamma \in (0, 1)$ ensures that this range is bounded within $(-1, 1)$.
Note that the convergence of $\rib$ leads to that of $\rub$. 
\end{proof}

\subsubsection{Complexity Analysis  (T3)}
\label{sec:proposed:analysis_complexity}

We analyze the time complexity of Algorithm~\ref{alg:method} for \method. 
Suppose $m \coloneqq \sum_{b \in \B} |\Eb|$ is the total number of interactions, and $n \coloneqq |\users| + |\items|$ is the number of nodes in each behavior graph.
Note that each bi-adjacency matrix $\Ab$ is sparse and is therefore stored in a sparse matrix format, such as compressed sparse row (CSR)~\cite{daglib}.

\begin{theorem}[Time Complexity]
The time complexity of Algorithm~\ref{alg:method} for \method is $O(K  (m + |\B| n))$, where $K$ is the maximum number of iterations and $|\B|$ is the number of behaviors.
\end{theorem}

\begin{proof}
Algorithm~\ref{alg:method} repeats the power iteration in its inner loop for each behavior, following the symmetric normalization and initialization, both of which require $O(|\Eb| + |\Ub| + |\Ib|)$ time.
The power iteration repeats at most $K$ times, with its main bottleneck arising from sparse matrix-vector multiplications~\cite{daglib}, requiring $O(K(|\Eb| + |\Ub| + |\Ib|))$ time for all iterations.
The outer loop repeats $|\B|$ times, as the length of the cascading chain is $|\B|$.
Putting everything together, the total time complexity is $\sum_{b \in \B} (K+1)(|\Eb| + |\Ub| + |\Ib|) = (K+1)(m + |\B| n) \in O(K  (m + |\B| n))$, where $|\Ub| = |\users|$ and $|\Ib| = |\items|$.
\end{proof}

Note that in real-world multi-behavior datasets, the main factor is the number $m$ of interactions, since $|\B|$ is small (e.g., 3 or 4), $K$ is a fixed constant, and the number $n$ of nodes is much smaller than $m$, as shown in Table~\ref{tab:data}. 
This indicates that the algorithm for \method primarily takes $O(m)$ time, i.e., it scales linearly with the number of interactions, as empirically verified in Figure~\ref{fig:scalability}.

%% file: 050experiment.tex
\section{Experiments}
\label{sec:experiments}
In this section, we conducted experiments to address the following questions:

\begin{enumerate}[leftmargin=9mm,noitemsep]
    \item[Q1.] {
        \textbf{Recommendation performance. } 
        How effective is the personalized ranking provided by \method for multi-behavior recommendation compared to its competitors?
    }
    \item[Q2.] {
        \textbf{Ablation study.} How does each module of \method affect its recommendation performance?
    }
    \item[Q3.] {
        \textbf{Effect of hyperparameters.} 
        How do the hyperparameters of \method, such as smoothness, query fitting, and cascading alignment, influence its performance?
    }
    \item[Q4.] {
        \textbf{Convergence.} 
        Does the iterative algorithm of \method converge as the number of iterations increases?
    }
    \item[Q5.] {
        \textbf{Computational Efficiency} Does the \method yield linear computational complexity w.r.t. the number of total interactions?
    }
\end{enumerate}

\def\arraystretch{1.1} 
\setlength{\tabcolsep}{12pt}
\begin{table}[t]
\small
\caption{
    Data statistics of multi-behavior interactions.
}
\label{tab:data}
\centering
\begin{tabular}{crrrrrr}
    \hline
    \toprule
    \textbf{Dataset} & Users & Items & Views& 
    \makecell[r]{Collects$^{*}$\\or Shares$^{\dagger}$}
     & \makecell[r]{Carts$^{*}$\\or Likes$^{\dagger}$} & \makecell[r]{Buys$^{*}$\\or Follows$^{\dagger}$}\ \\
    \midrule
    \texttt{Taobao}$^{*}$ & 15,449 & 11,953 & 873,954 & - & 195,476  & 92,180 \\
    \texttt{Tenrec}$^{\dagger}$ & 27,948 & 15,440 & 1,489,997  & 13,947  & 1,914  & 1,307  \\
    \texttt{Tmall}$^{*}$ & 41,738 & 11,953 & 1,813,498  & 221,514 & 1,996  & 255,586 \\
    \bottomrule
    \hline
\end{tabular}
\begin{tablenotes}[flushleft]\footnotesize
{\item[] $\dagger$: Note that \tenrec includes shares, likes, and follows as corresponding behaviors.}
\end{tablenotes}
\end{table}

\subsection{Experimental Setting}
\label{sec:exp:setting}
We describe the setup for our experiments on multi-behavior recommendation.

\smallsection{Datasets}
We conducted experiments on three real-world multi-behavior datasets:  
\taobao~\cite{MengZGGLZZTLZ23, LeeLKSJ24xxaw}, \tenrec~\cite{abs-2210-10629, ZhangBCSYGWHH24} and \tmall~\cite{YanCGSLSL24, MengZGGLZZTLZ23} which are publicly available and have been used as standard benchmarks in previous studies~\cite{LeeLKSJ24xxaw, MengZGGLZZTLZ23, YanCGSLSL24, ZhangBCSYGWHH24, YanYCSSP23}.
The detailed statistics of the datasets are summarized in Table~\ref{tab:data}.
The \texttt{Tmall} dataset has four behavior types: \textit{view}, \textit{add-to-collect}, \textit{add-to-cart}, and \textit{buy}, while \texttt{Taobao}, apart from add-to-collect, consists of three types.
The behavior types of \tenrec are \textit{view}, \textit{share}, \textit{like} and \textit{follow}.
Following the previous studies~\cite{LeeLKSJ24xxaw, MengZGGLZZTLZ23, ZhangBCSYGWHH24}, we set the target behavior as \textit{buy} for $\{\taobao, \tmall\}$ and \textit{follow} for $\{\tenrec\}$, and preprocessed duplicate interactions by retaining only the earliest occurrence for each behavior.

\smallsection{Competitors}
We compared \method with traditional graph ranking methods such as \textbf{RWR}~\cite{TianJ13}, \textbf{CoHITS}~\cite{DengDLK09kqpg}, \textbf{BiRank}~\cite{HeHGKW17}, as well as single-behavior representation learning methods such as \textbf{LightGCN}~\cite{HeHDWLZW20} and \textbf{MF-BPR}~\cite{RendleRFGS12}.
Note that these methods were designed to operate on a single behavior graph. 
For these, we combined all behavior graphs into a unified graph\footnote{We performed an element-wise union operation across the adjacency matrices of all behaviors to construct the adjacency matrix of the unified graph.}, on which these methods were applied. 
Results based solely on the target behavior interactions are provided in~\ref{sec:appendix:single}.
We further compare \method with \textbf{NRank}, an extended version of \textbf{BiRank}~\cite{HeHGKW17} for multi-partite graphs, where the interactions of each behavior form a subgraph between users and items. In this graph, the set of users is shared across behaviors, while the set of items is distinct for each behavior.
We also compared our method with state-of-the-art representation learning methods for multi-behavior recommendation, including: 1) non-cascading approaches such as \textbf{MB-HGCN}~\cite{YanYCSSP23} and \textbf{MuLe}~\cite{LeeLKSJ24xxaw}, and 2) cascading approaches such as \textbf{PKEF}~\cite{MengMZYZL23} and \textbf{HEC-GCN}~\cite{yin2024hecgcn}.
Finally, we evaluated our method against \textbf{BPMR}~\cite{LiLCYLLD24}, a pattern-mining-based method that achieves state-of-the-art accuracy in the recommendation task.

\smallsection{Training and evaluation protocol}
We follow the leave-one-out setting, which has been broadly used in previous studies~\cite{MengZGGLZZTLZ23,LeeLKSJ24xxaw,ChengCHLZGP23fqvn}, where the test set consists of the last interacted item for each users. 
The second most recently interacted item for each user forms the validation set for tuning hyperparameters, while the remaining interactions are used for training.
In the evaluation phase, items within the test set for each user are ranked based on ranking or predicted scores by models. where its top-$k$ ranking quality is measured by HR@$k$ and NDCG@$k$~\cite{MengZGGLZZTLZ23,LeeLKSJ24xxaw,ChengCHLZGP23fqvn,ZhangBCSYGWHH24}. Note that target behavior items(i.e., buy) in the training interaction are excluded during testing.
HR@$k$ measures how often relevant items, on average, appear in the recommendation for each user.
NDCG@$k$ considers both relevance and order of relevant items in a ranking, averaged across all users.

\smallsection{Hyperparameter tuning}
For each dataset, we conducted a grid search to tune hyperparameters on the validation set and reported the test performance with the validated hyperparameters. 
The hyperparameters $\alpha$ and $\beta$ of \method are tuned in $[0,1]$ such that $0 \leq \alpha + \beta \leq 1$.
The validated values of $\alpha$ and $\beta$ for each dataset are provided in Appendix~\ref{sec:appendix:Effect of Hyperparameters_detail}.
For the cascading sequence $\mathcal{C}$, we follow a natural sequence of user behaviors based on their semantics, as used in previous studies~\cite{yin2024hecgcn, MengMZYZL23, LiuXWY00024, YanCGSLSL24, ChengCHLZGP23fqvn}.
We set $\mathcal{C}$ to $(\texttt{view}\rightarrow\texttt{collect}\rightarrow\texttt{cart}\rightarrow\texttt{buy})$ for \tmall, $(\texttt{view}\rightarrow\texttt{cart}\rightarrow\texttt{buy})$ for \taobao, and $(\texttt{view}\rightarrow\texttt{share}\rightarrow\texttt{like}\rightarrow\texttt{follow})$ for \tenrec (see Appendix~\ref{sec:appendix:Performance_for_Different_Cascading_Sequences} for results with different cascading sequences).
For each competitor, we followed the range of its hyperparameters, as reported in the corresponding paper.

\smallsection{Machine and implementation}
We used a workstation with AMD 5955WX and RTX 4090 (24GB VRAM).
Our method \method was implemented using Pytorch 2.0 in Python 3.9. 
Note that the BMPR's algorithm~\cite{LiLCYLLD24} is designed to run sequentially on a CPU and it is hard to parallelize, while other methods are easily parallelizable due to matrix operations and execute on a GPU.
For the other methods, we used their open-source implementations, with detailed information provided in Appendix~\ref{sec:appendix:competitors:information}.

\def\arraystretch{1.1} 
\setlength{\tabcolsep}{6pt}
\begin{table}[t]
\caption{
Multi-behavior recommendation performance in terms of HR@10 and NDCG@10, with the best results in bold and the second-best results underlined.
`\% impv.’ indicates the percentage improvement of the best model over the second-best model.
RL stands for representation learning, GR for graph ranking, and PM for pattern-mining methods.
\textbf{Our \method shows the best recommendation performance among all the tested methods across the datasets.}
}
\small
\label{tab:performance}
\centering
\begin{tabular}{cc|ccc|ccc}
\hline
\toprule
\multirow{2}{*}{\bf Methods} & \multirow{2}{*}{\bf Type} & \multicolumn{3}{c|}{\bf HR@10}       & \multicolumn{3}{c}{\bf NDCG@10}      \\
             & & \taobao  & \tenrec  & \tmall & \taobao  & \tenrec   & \tmall  \\ \midrule
MF-BPR  & RL & 0.0758  & 0.1244 & 0.0855 & 0.0387  & 0.0575  & 0.0423  \\
LightGCN & RL & 0.1025  & 0.1069 & 0.1162 & 0.0566  & 0.0526  & 0.0625  \\ \midrule
MB-HGCN & RL       & 0.1261  & 0.1133 & 0.1413 & 0.0666  & 0.0618  & 0.0753  \\ 
MuLe   & RL        & 0.1949  & 0.1920  & 0.2097 & 0.1128  & 0.1100  & 0.1175  \\ \midrule
PKEF   & RL        & 0.1349  & 0.0968 & 0.1222 & 0.0763  & 0.0530  & 0.0696  \\
HEC-GCN & RL       & 0.1905  & 0.2673 & 0.1806 & 0.1038  & 0.1565  & 0.1000  \\ \midrule
RWR & GR     & 0.2130  & 0.2074 & 0.2712 & 0.0988  & 0.0962  & 0.1284  \\
CoHITS & GR  & 0.2128  & 0.2074 & 0.2713 & 0.0988  & 0.0957  & 0.1284  \\
BiRank & GR   & \underline{0.3034}  & 0.2949 & \underline{0.3550} & \underline{0.1517}  & 0.1257  & \underline{0.1819}  \\ 
NRank & GR         & 0.2989  & \underline{0.4562} & 0.3477 & 0.1419  & 0.2508  & 0.1726  \\
\midrule
BPMR  & PM        & 0.2846  & 0.4286 & 0.3289 & 0.1429  & \underline{0.2579}  & 0.1598  \\ \midrule
\textbf{\method} & GR  & \textbf{0.3324}  & \textbf{0.4747} & \textbf{0.3751} & \textbf{0.1626}  & \textbf{0.2723}  & \textbf{0.1871}  \\ \midrule
\% impv.   & -     & 9.56\%  & 5.67\% & 5.67\% & 7.16\%  & 5.57\%  & 2.85\%  \\
\bottomrule
\hline
\end{tabular} 
\end{table}

\subsection{Recommendation Performance (Q1)} 
We evaluate the effectiveness of \method in multi-behavior recommendation by comparing it to its competitors.

\smallsection{Top-10 performance}
We first examine the recommendation quality of the top 10 highest-ranked items by each method, measured in terms of HR@10 and NDCG@10.
From Table~\ref{tab:performance}, we observe the following:
\begin{itemize}[leftmargin=9mm,noitemsep]
    \item{
        Our proposed \method achieves superior performance compared to its competitors across all datasets, providing \textbf{up to 9.56\% improvement in HR@10 and up to 7.16\% in NDCG@10 over the second-best model}, with the most notable gains observed on the \taobao dataset.
    }
    \item{
        The representation learning (RL) methods generally underperform compared to the graph ranking (GR) or the pattern-mining (PM) methods.
        Specifically, naive GR models on a unified graph such as BiRank surpass state-of-the-art RL models exploiting multi-behavior interactions such as MuLe and HEC-GCN on \taobao and \tmall, highlighting the RL approach’s limitations in providing accurate recommendations due to their limited expressiveness, mainly caused by over-smoothing and bias issues.
    }
    \item{
        For the GR methods, exploiting a cascading sequence in measuring ranking scores, as in \method, is beneficial because it outperforms other GR methods, such as BiRank and NRank, which use all interactions but do not fully exploit the semantics of behaviors.
    }
    \item{
        It matters more to precisely encode embeddings than to use additional information, such as the cascading pattern, in the RL methods. 
        As evidence of this, HEC-GCN, a cascading method, performs better than MuLe, a non-cascading method, on \tenrec, while the opposite is observed on the other datasets.
    }
    \item {
        Our \method outperforms BPMR, which is recognized for achieving state-of-the-art accuracy and surpassing the RL methods. 
        The main difference is that BPMR considers only a few steps of paths, while our method accounts for the global structure of each behavior graph and the cascading effect, enabling it to generate higher-quality scores.      Furthermore, BPMR is significantly slower than \method, as discussed in Section~\ref{sec:exp:efficiency}.
    }
\end{itemize}

\begin{figure}[t]
    \centering
    \includegraphics[width=0.95\linewidth]{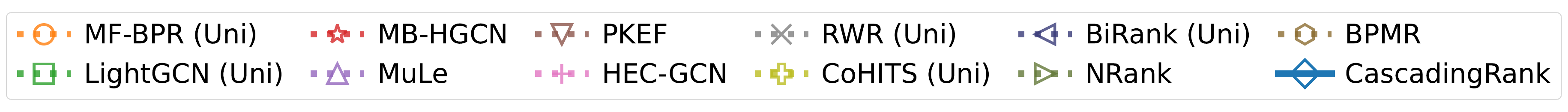}\\
    \subfigure[\taobao]{
        \hspace{-7mm}
        \includegraphics[width=0.318\linewidth]{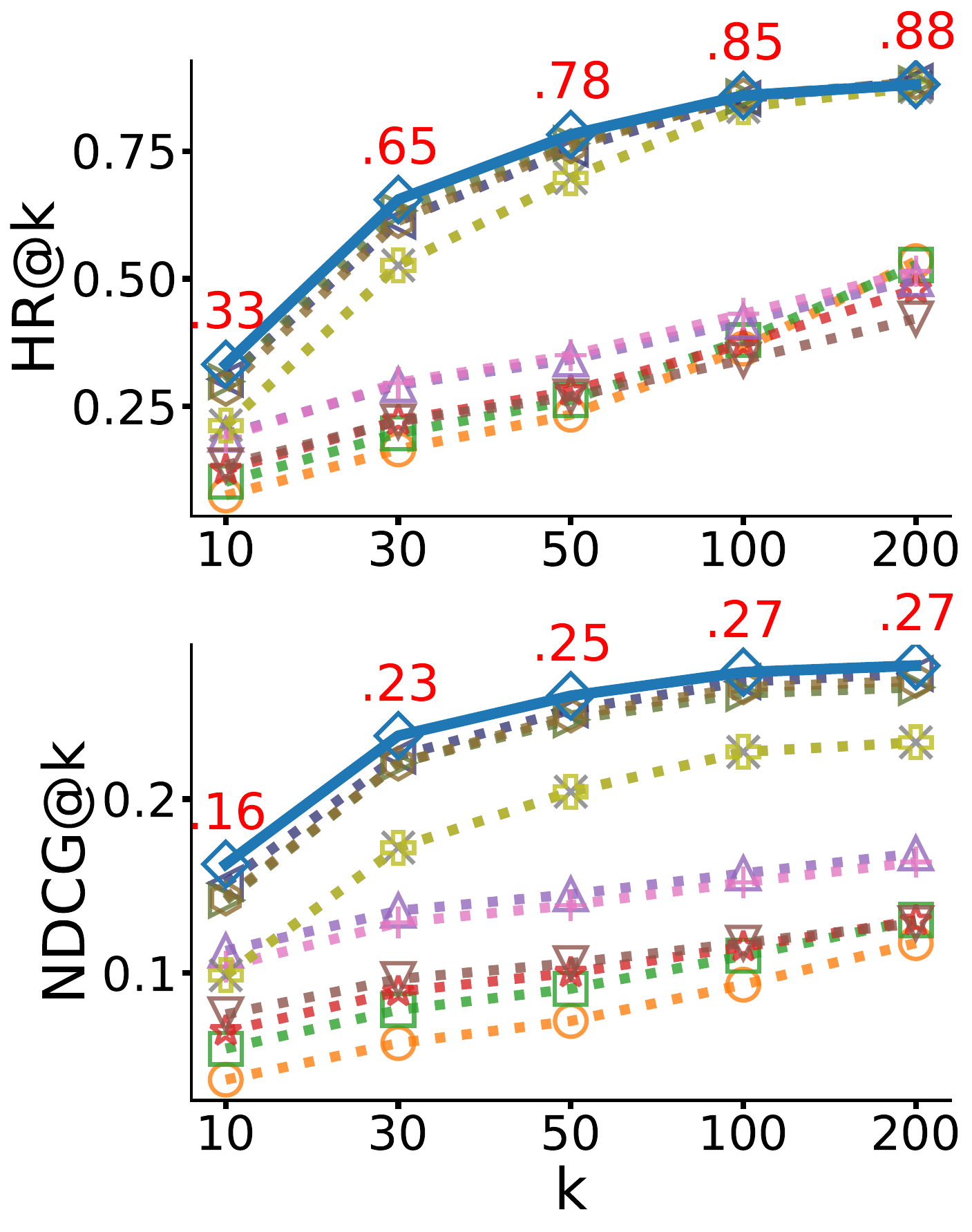}
        \label{fig:appendix:All_k_score:Taobao}
    }
    \subfigure[\tenrec]{
        \hspace{-3mm}
        \includegraphics[width=0.300\linewidth]{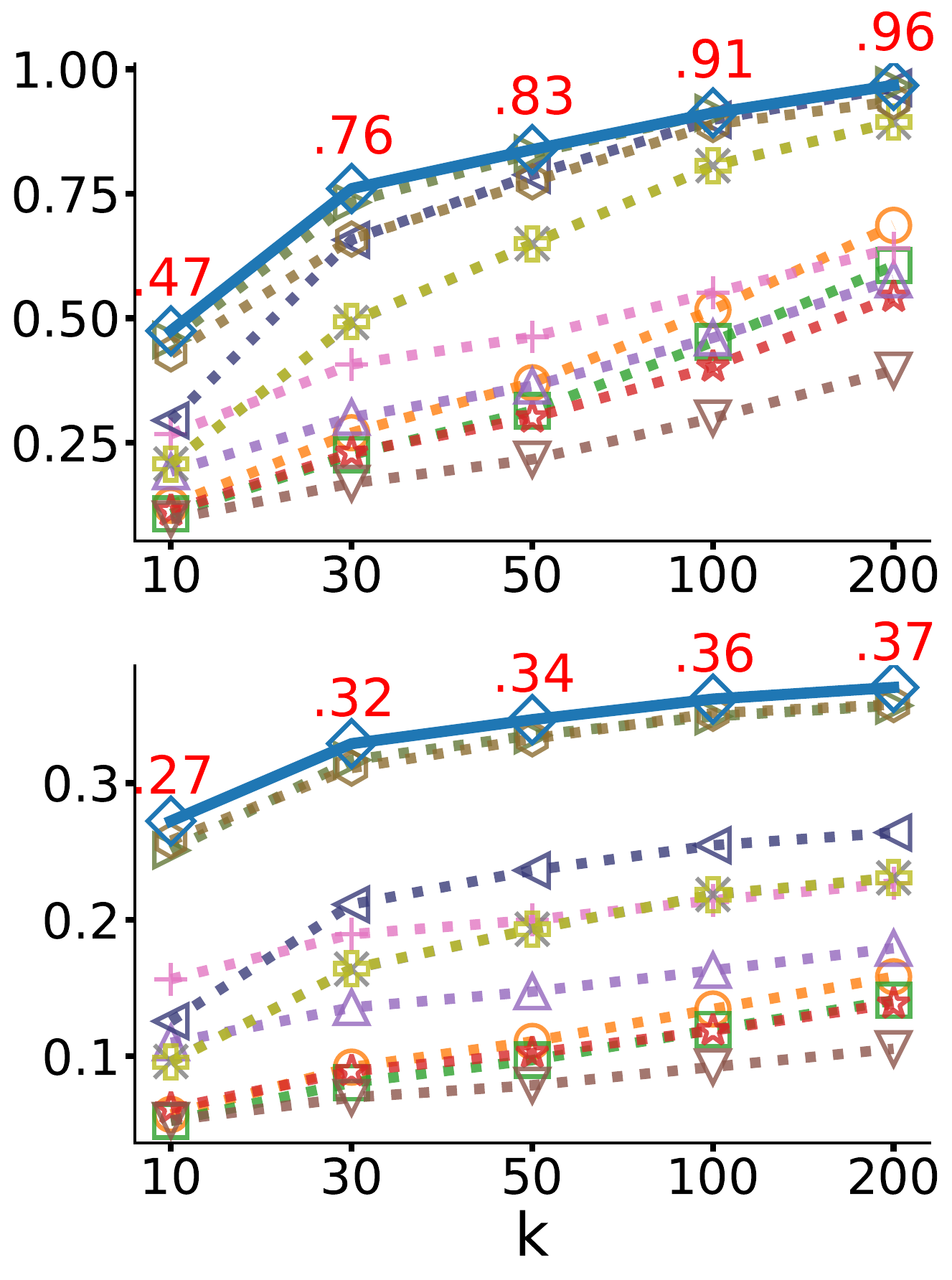}
        \label{fig:experiments:All_k_score:Tenrec}
    }
    \subfigure[\tmall]{
        \hspace{-4.5mm}
        \includegraphics[width=0.30\linewidth]{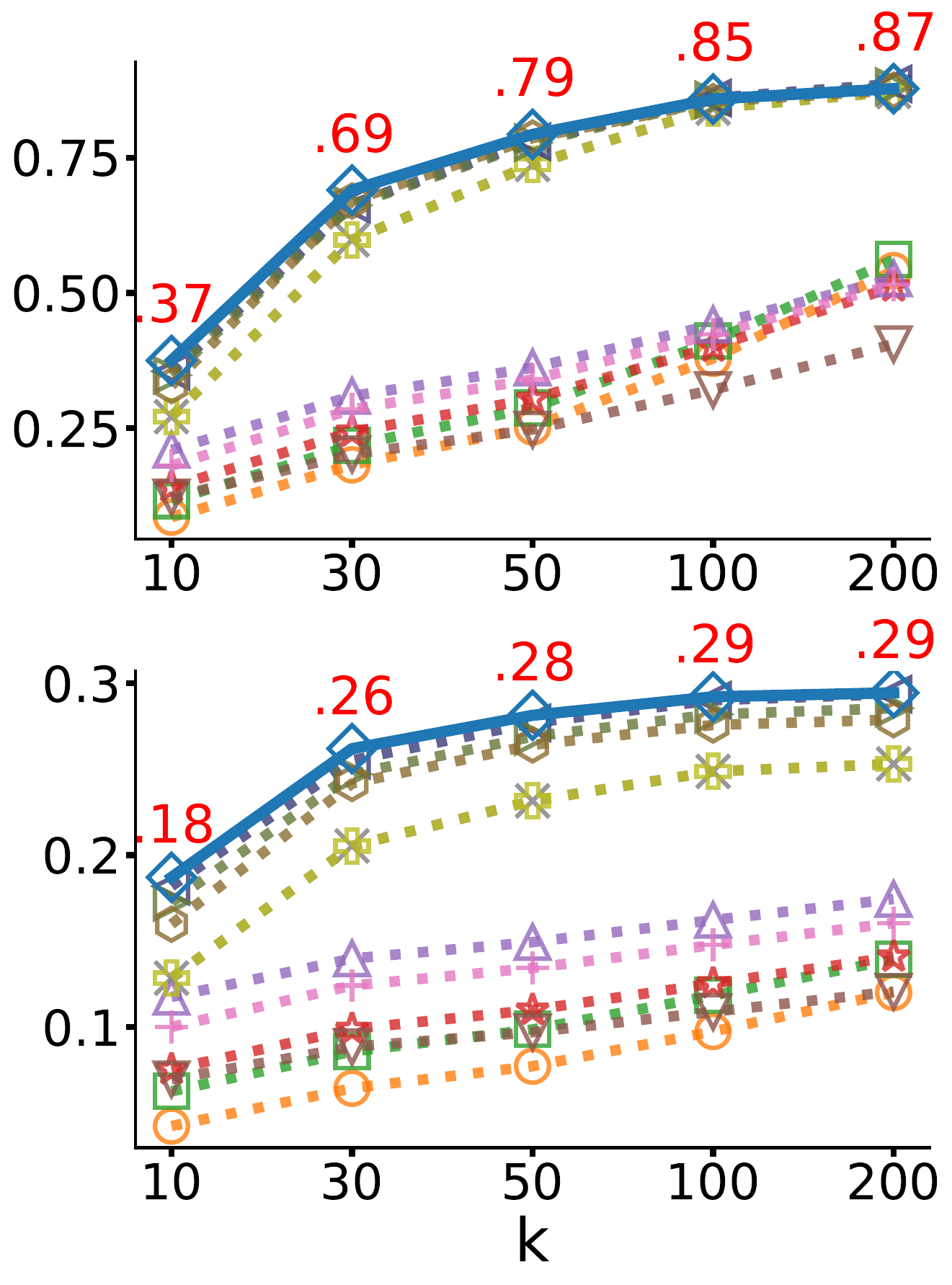}
        \label{fig:appendix:All_k_score:Tmall}
    }
    
    \caption{
        \label{fig:experiments:varik}
        Recommendation performance in HR@$k$ and NDCG@$k$, where $k$ varies in \{10, 30, 50, 100, 200\}.
        \textbf{Our \method provides better ranking scores than its competitors in both metrics.}
    }
\end{figure}

\smallsection{Top-$k$ performance}
We further evaluate the ranking quality of \method by comparing it with its competitors in terms of HR@$k$ and NDCG@$k$ for various values of $k$ ranging from 10 to 200.
As shown in Figure~\ref{fig:experiments:varik}, \textbf{our \method achieves the highest ranking quality for items that test users are likely to purchase}, outperforming its competitors across all datasets.
Specifically, NDCG@$k$ of \method is the highest across all $k$, significantly outperforming the RL methods, indicating that the target items are more likely to be ranked higher in our results.
It is worth noting that HR@$k$ of the GR methods increases significantly as $k$ grows, compared to that of the RL methods, indicating that the rankings produced by the GR methods are more likely to contain the target items.
This implies that the GR methods are more suitable for generating candidate items, which can then be refined through re-ranking process~\cite{ValizadeganJZM09}.

\begin{figure}[t]
    \centering
    \includegraphics[width=1\linewidth]{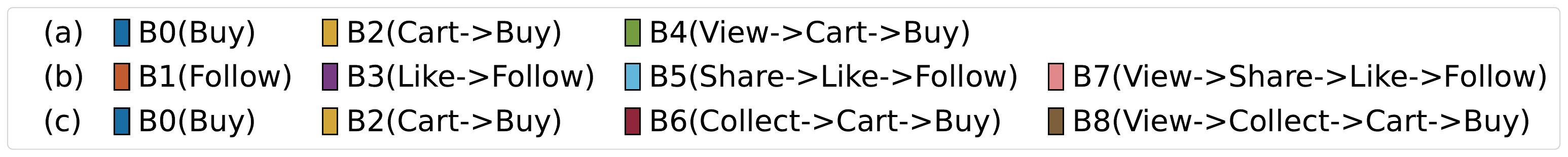}
    
    \subfigure[\taobao]{
        \hspace{-3mm}
        \includegraphics[width=0.320\linewidth]{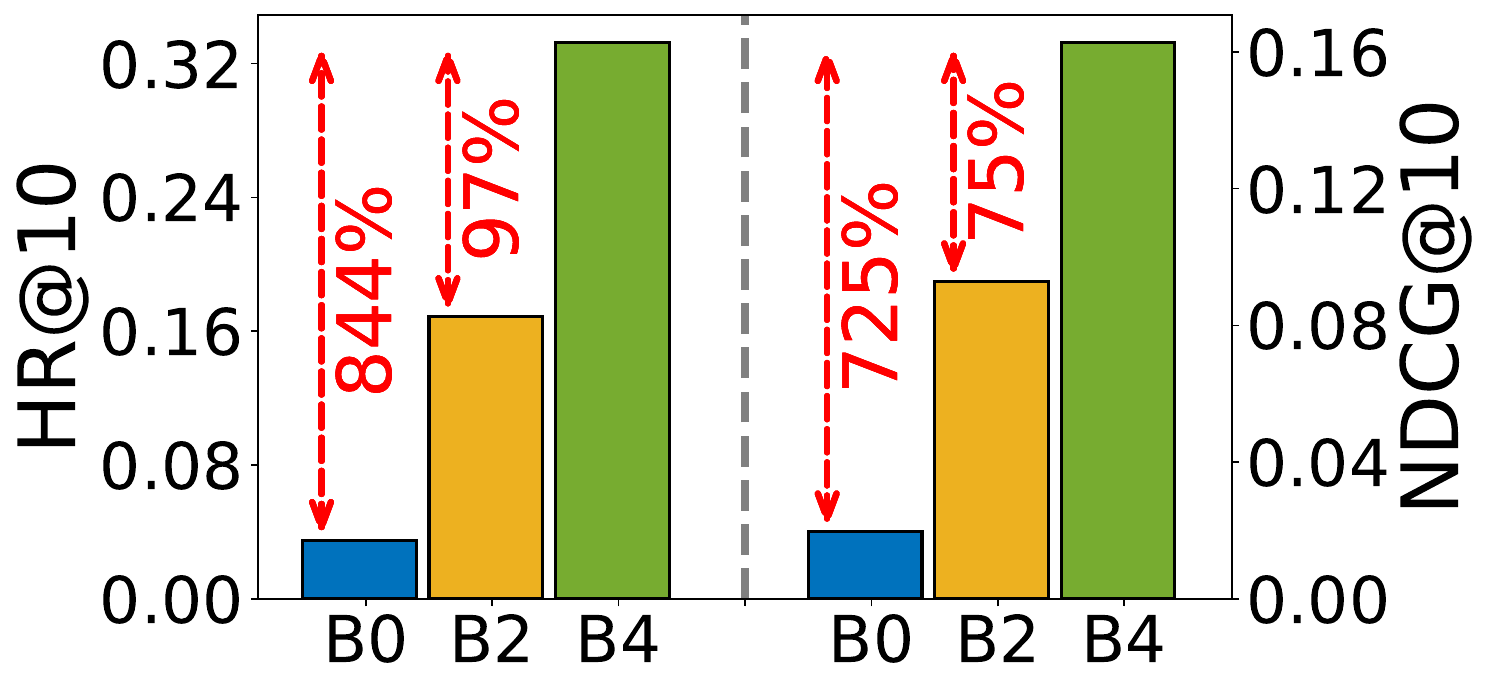}
        \label{fig:experiments:ablation:Taobao}
    }
    \subfigure[\tenrec]{
        \hspace{-3mm}
        \includegraphics[width=0.312\linewidth]{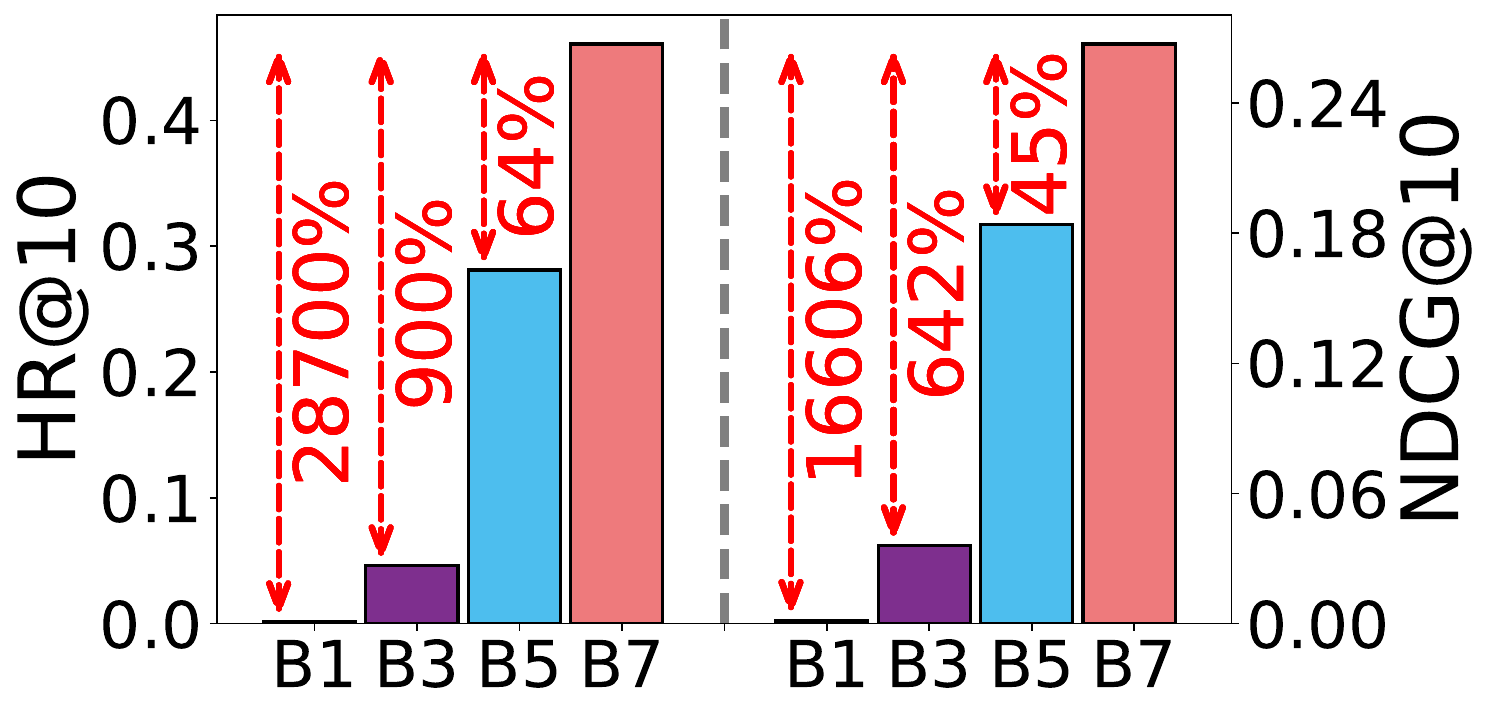}
        \label{fig:experiments:ablation:Tenrec}
    }
    \subfigure[\tmall]{
        \hspace{-3mm}
        \includegraphics[width=0.320\linewidth]{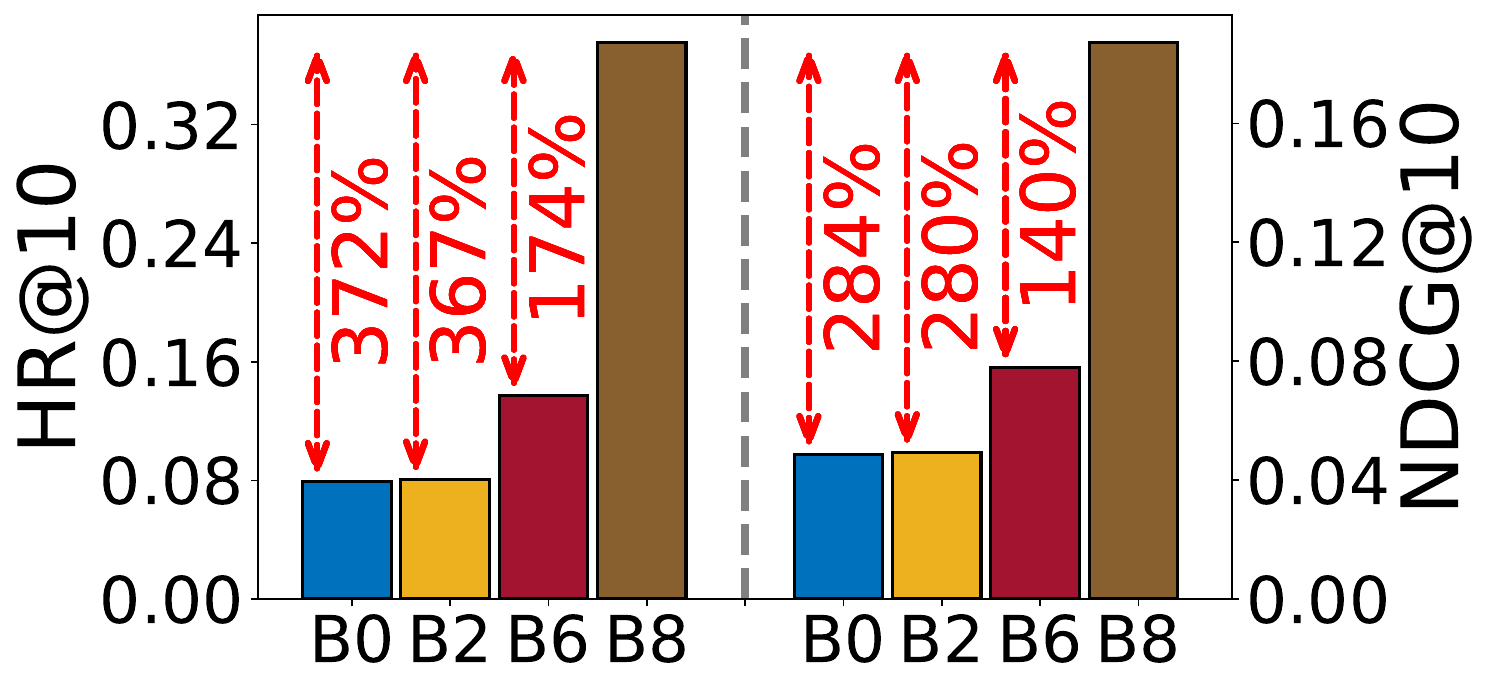}
        \label{fig:experiments:ablation:Tmall}
        \hspace{-3mm}
    }
    \caption{
        \label{fig:experiments:ablationstudy}
        Effect of behaviors in the cascading sequence, where \texttt{B4} is used for \method, \texttt{B7} for \tenrec, and \texttt{B8} for \tmall to produce the final ranking scores of \method.
        \textbf{Note that utilizing all behaviors in the sequence is beneficial for recommendation}, as the performance degrades when earlier behaviors are excluded from the sequence.
    }
\end{figure}

\subsection{Ablation Study (Q2)}
We investigate the effectiveness of our design choices in \method through ablation studies.

\smallsection{Effect of auxiliary behaviors in the cascading sequence}
We conducted an ablation study to verify the impact of auxiliary behaviors in the cascading sequences used in our method for each dataset.
For this experiment, we sequentially excluded each auxiliary behavior from the cascading sequence $\mathcal{C}$, where it was initially set to $\texttt{B8}: (\texttt{view}\rightarrow\texttt{collect}\rightarrow\texttt{cart}\rightarrow\texttt{buy})$ for \tmall, $\texttt{B4}:(\texttt{view}\rightarrow\texttt{cart}\rightarrow\texttt{buy})$ for \taobao, and $\texttt{B7}:(\texttt{view}\rightarrow\texttt{share}\rightarrow\texttt{like}\rightarrow\texttt{follow})$ for \tenrec.
As shown in Figure~\ref{fig:experiments:ablationstudy}, using all behaviors in the cascading sequence leads to better recommendations than the variants that exclude auxiliary behaviors.
Specifically, using only the graph of interactions of the target behavior shows the worst performance across all datasets, and the performance improves as auxiliary behaviors are added in the order of the cascading sequence.
This indicates that incorporating all auxiliary behaviors in the order of the cascading sequence is essential for achieving optimal performance.

\def\arraystretch{1.1} 
\setlength{\tabcolsep}{6.7pt}
\begin{table}[t]
\caption{
Effect of normalization on measuring the ranking scores of \method.
\textbf{Ranking scores with symmetric normalization provides more accurate recommendation than those with column normalization.}
}
\label{tab:ablation:normalization}
\centering
\small
\begin{tabular}{c|ccc|ccc}
\hline
\toprule
\multirow{2}{*}{\bf Variants} & \multicolumn{3}{c|}{\bf HR@10}       & \multicolumn{3}{c}{\bf NDCG@10}      \\
             & \taobao  & \tenrec  & \tmall & \taobao  & \tenrec   & \tmall  \\ 
\midrule
\methodcol       & 0.2919  & 0.4439 & 0.3270  & 0.1465  & 0.2559  & 0.1605  \\
\methodsym       & \bf 0.3324  & \bf 0.4747 & \bf 0.3751  & \bf 0.1626  & \bf 0.2723  & \bf 0.1871  \\
\midrule
\% impv.                 & 13.86\% & 6.93\% & 14.74\% & 10.97\%  & 6.39\%  & 16.60\% \\ 
\bottomrule
\hline
\end{tabular} 
\end{table}

\smallsection{Effect of normalization}
We check the effect of normalization for estimating ranking scores in Equation~\eqref{eq:cascrank:vect}. 
For this experiment, we compared the following:
\begin{itemize}[leftmargin=9mm,noitemsep]
    \item {
        \methodcol: it uses column normalization on each adjacency matrix, i.e., $\Abnorm = \Ab \Dib^{-1}$ and $\Abnorm^{\top} = \Ab^{\top} \Dub^{-1}$.
    }
    \item {
        \methodsym: it uses symmetric normalization on each adjacency matrix, i.e., $\Abnorm = \Dub^{-1/2} \Ab \Dib^{-1/2}$ and $\Abnorm^{\top} =  \Dib^{-1/2} \Ab^{\top} \Dub^{-1/2}$.
    }
\end{itemize}
As shown in Table~\ref{tab:ablation:normalization}, the symmetric normalization achieves better recommendation performance than the column normalization, with improvements of up to 14.74\% in HR@$10$ and 16.60\% in NDCG@$10$ on the \tmall dataset.
This indicates that reducing the impact of both users and items in terms of size is more beneficial than reducing that of either one alone for scoring, especially when recommending items in long-tail distributions.

\begin{figure}[t]
    \centering
    \subfigure[$\alpha$: query fitting]{
        \includegraphics[width=0.3145\linewidth]{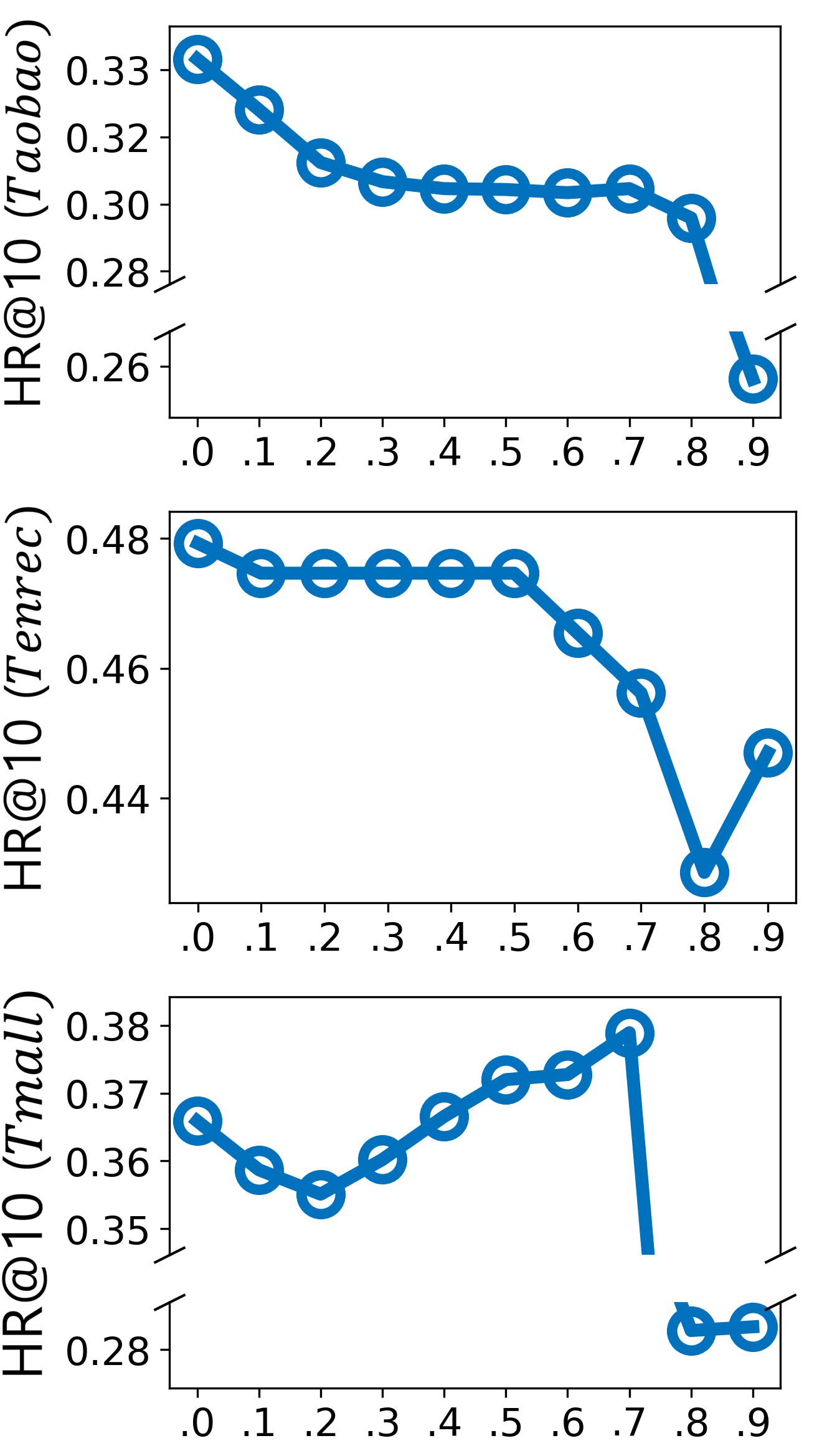}
        \label{fig:experiments:hyper_sens_split_view:alpha}
        \hspace{-4mm}
    }
    \subfigure[$\beta$: cascading alignment]{
        \includegraphics[width=0.3\linewidth]{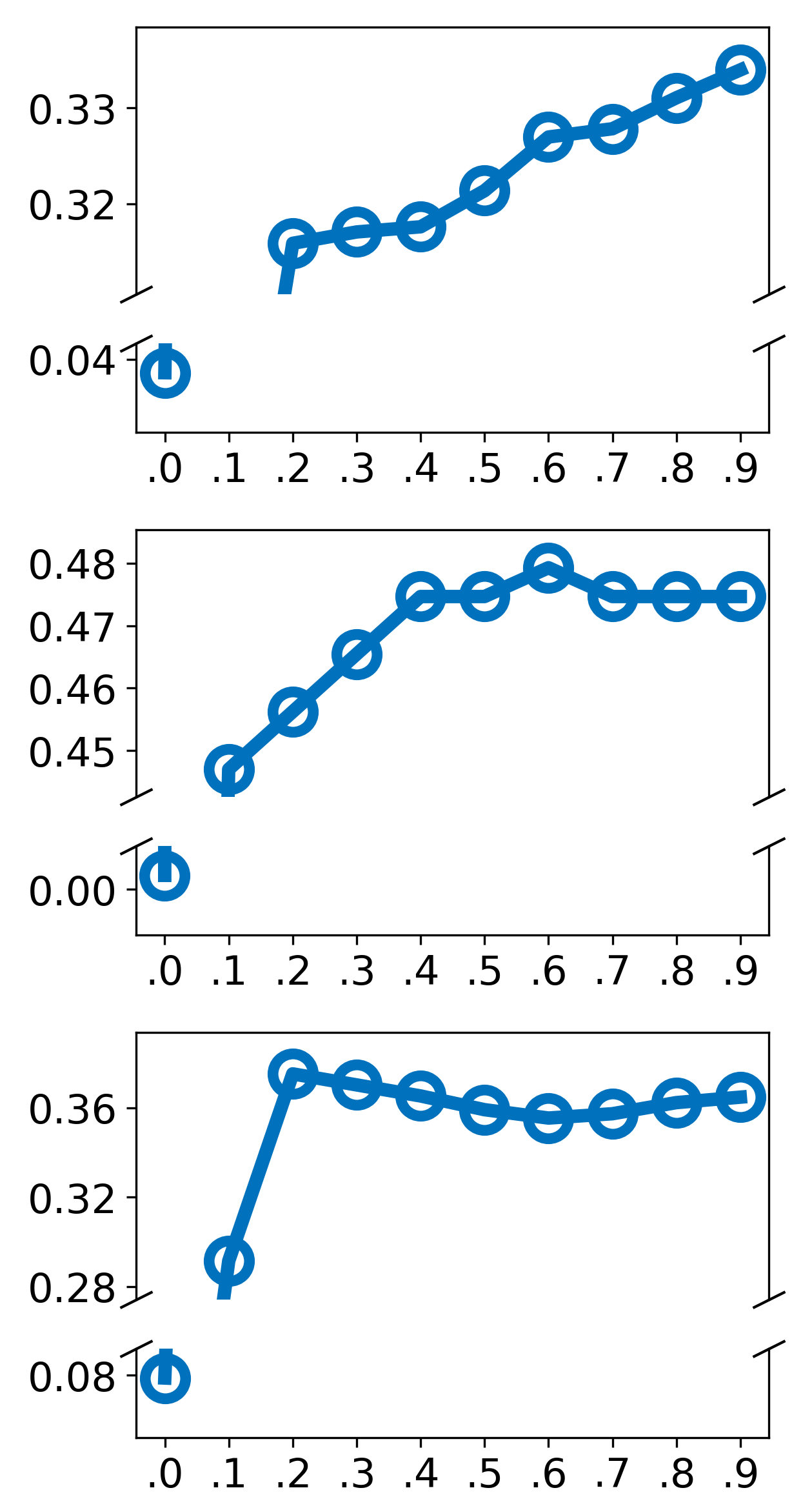}
        \label{fig:experiments:hyper_sens_split_view:beta}
        \hspace{-4mm}
        
    }
    \subfigure[$\gamma$: smoothing]{
        \includegraphics[width=0.3\linewidth]{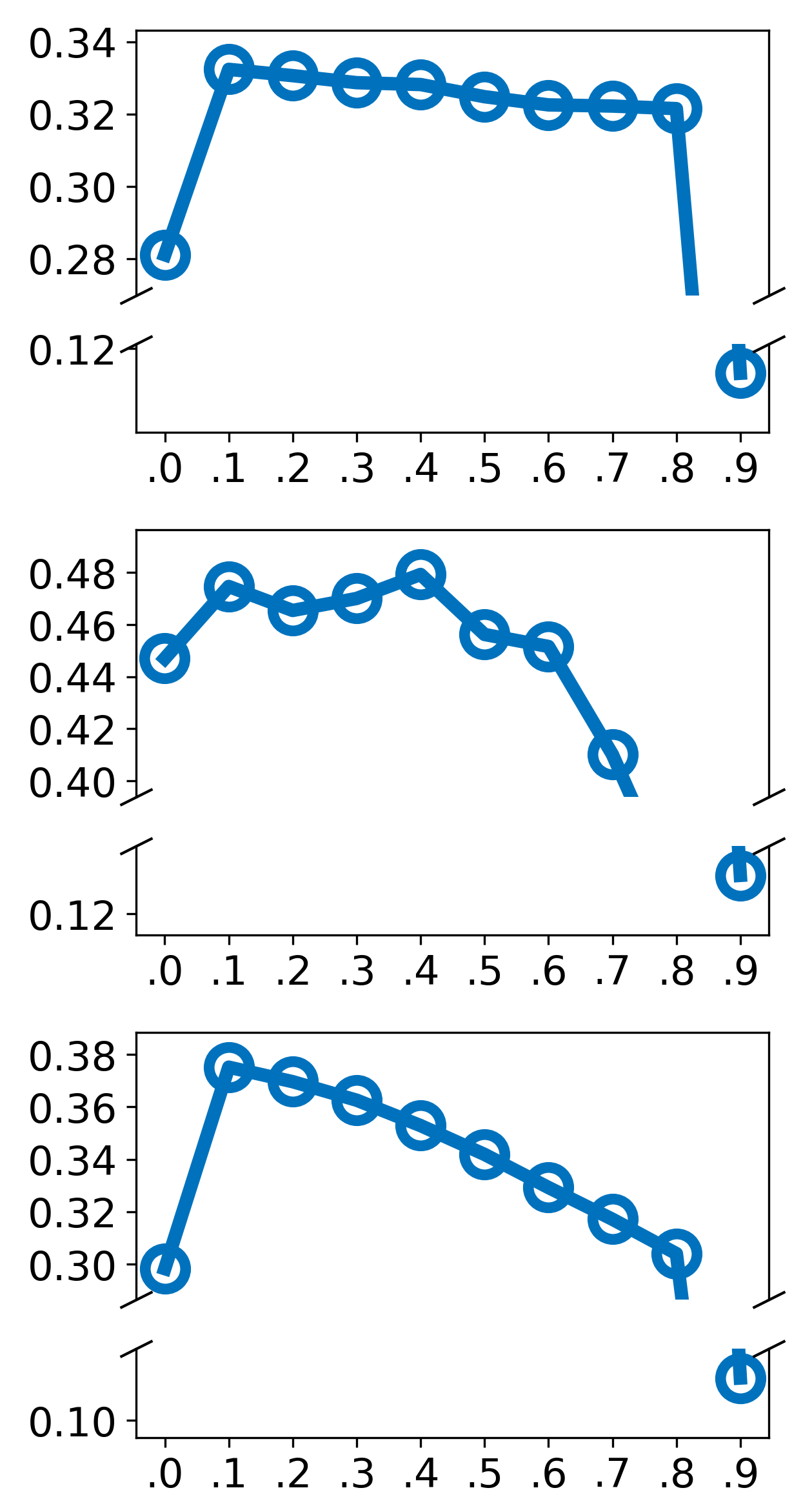}
        \label{fig:experiments:hyper_sens_split_view:gamma}
    }
    \caption{
    \label{fig:experiments:hyper_sens_split_view}
    Effect of the hyperparameters $\alpha$, $\beta$, and $\gamma$ of \method on the recommendation performance in HR@10, where $\gamma = 1 - \alpha - \beta$ is the strength of smoothing, and $\alpha$ and $\beta$ are the strengths of query fitting and cascading alignment, respectively.
    }
\end{figure}

\subsection{Effect of Hyperparameters (Q3)}
\label{sec:experiments:hyper_sens_split_view}
We investigate the impact of the hyperparameters $\alpha$, $\beta$, and $\gamma$ in \method on the performance of multi-behavior recommendation.
For this experiment, we varied $\alpha$ and $\beta$ from 0 to 1 with a step size of 0.1, and measured HR@10 of \method with them such that $\alpha + \beta$ is between 0 and 1, where the results of all possible combinations are provided in Appendix~\ref{sec:appendix:details:hyperparams}.
For each value of a hyperparameter\footnote{For better visualization, we excluded the results when the value of each hyperparameter is $1$, as the accuracies were significantly low in all cases.}, varied in increments of 0.1, we reported the maximum accuracy it achieved along with the possible values of the others, to analyze its effect.

Figure~\ref{fig:experiments:hyper_sens_split_view} shows the effects of $\alpha$, $\beta$, and $\gamma$ on the recommendation performance.
As shown in Figure~\ref{fig:experiments:hyper_sens_split_view:beta}, the accuracy improves with an increase in the strength $\beta$ of cascading alignment across all tested datasets, highlighting the importance of leveraging cascading information\footnote{However, relying solely on the cascading information, such as setting $\beta = 1$, results in poor performance, as shown in Figure~\ref{fig:appendix:Hyper-sens-detail}.}.
The effects of $\alpha$ and $\gamma$ depend on the datasets. 
For the strength $\alpha$ of query fitting, a smaller value works better on \taobao and \tenrec, while a moderately large value performs better on \tmall.
For the strength $\gamma$ of smoothing, the accuracy remains relatively high for values between 0.1 and 0.7 on \taobao and \tenrec, whereas the accuracy significantly drops after 0.1. 
This indicates that query information is far more crucial on \tmall compared to \taobao and \tenrec, whereas the smoothing plays a more significant role on the latter datasets.

\begin{figure}[t]
    \centering
    \includegraphics[width=0.7\linewidth]{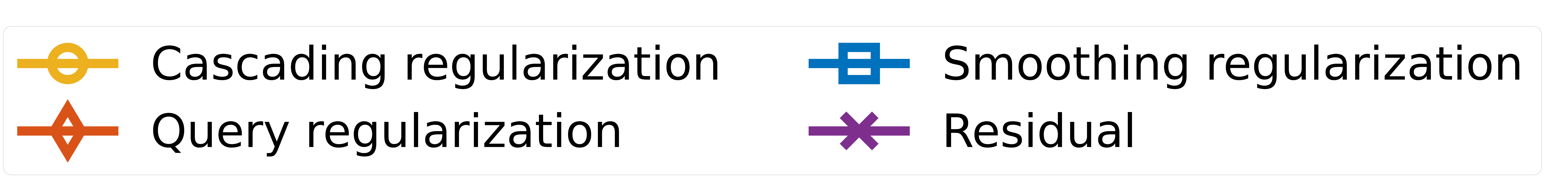}\\
    \subfigure[\taobao]{
        \includegraphics[width=0.32\linewidth]{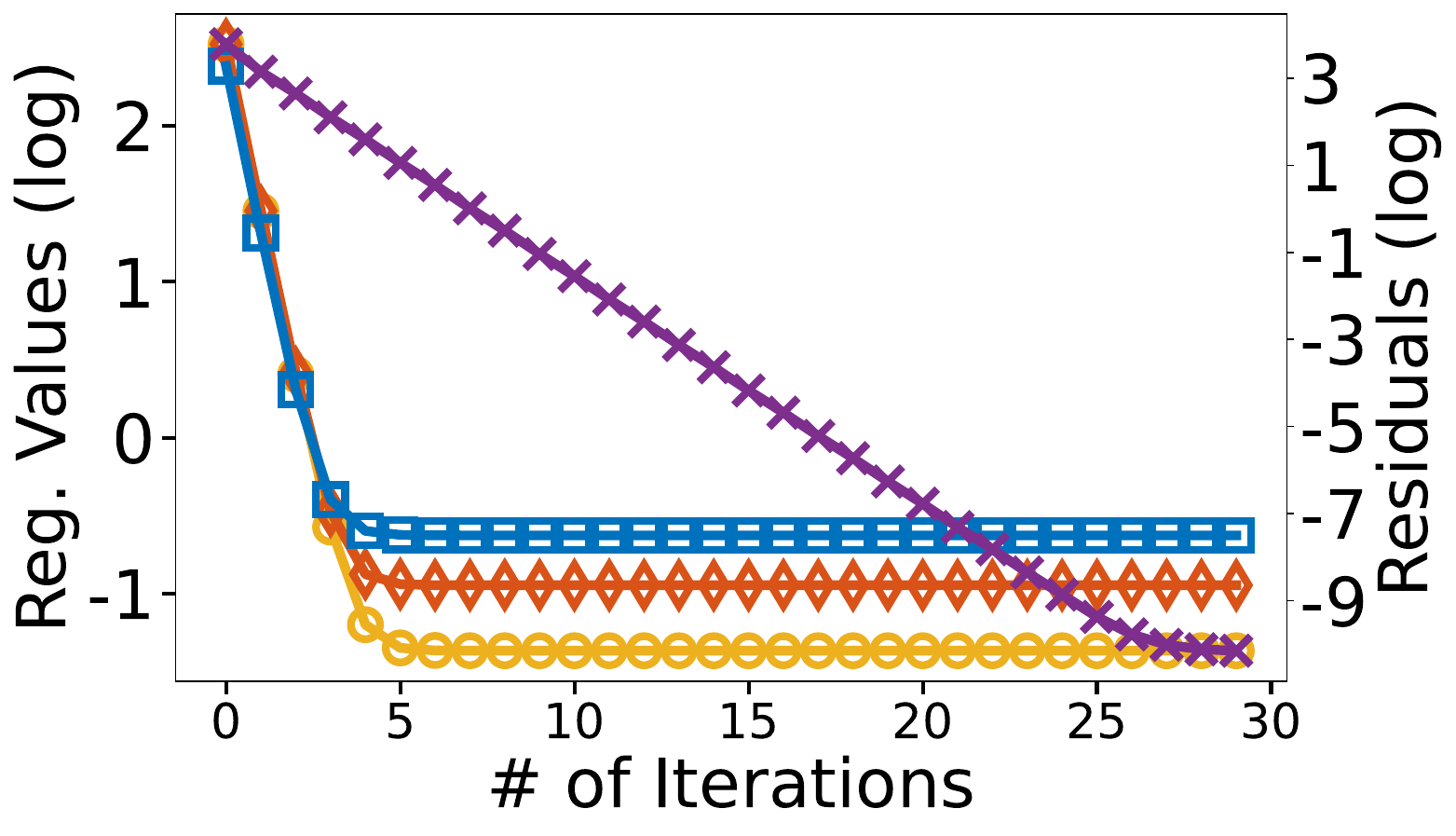}
        \hspace{-1mm}
    }
    \subfigure[\tenrec]{
        \hspace{-3mm}
        \includegraphics[width=0.320\linewidth]{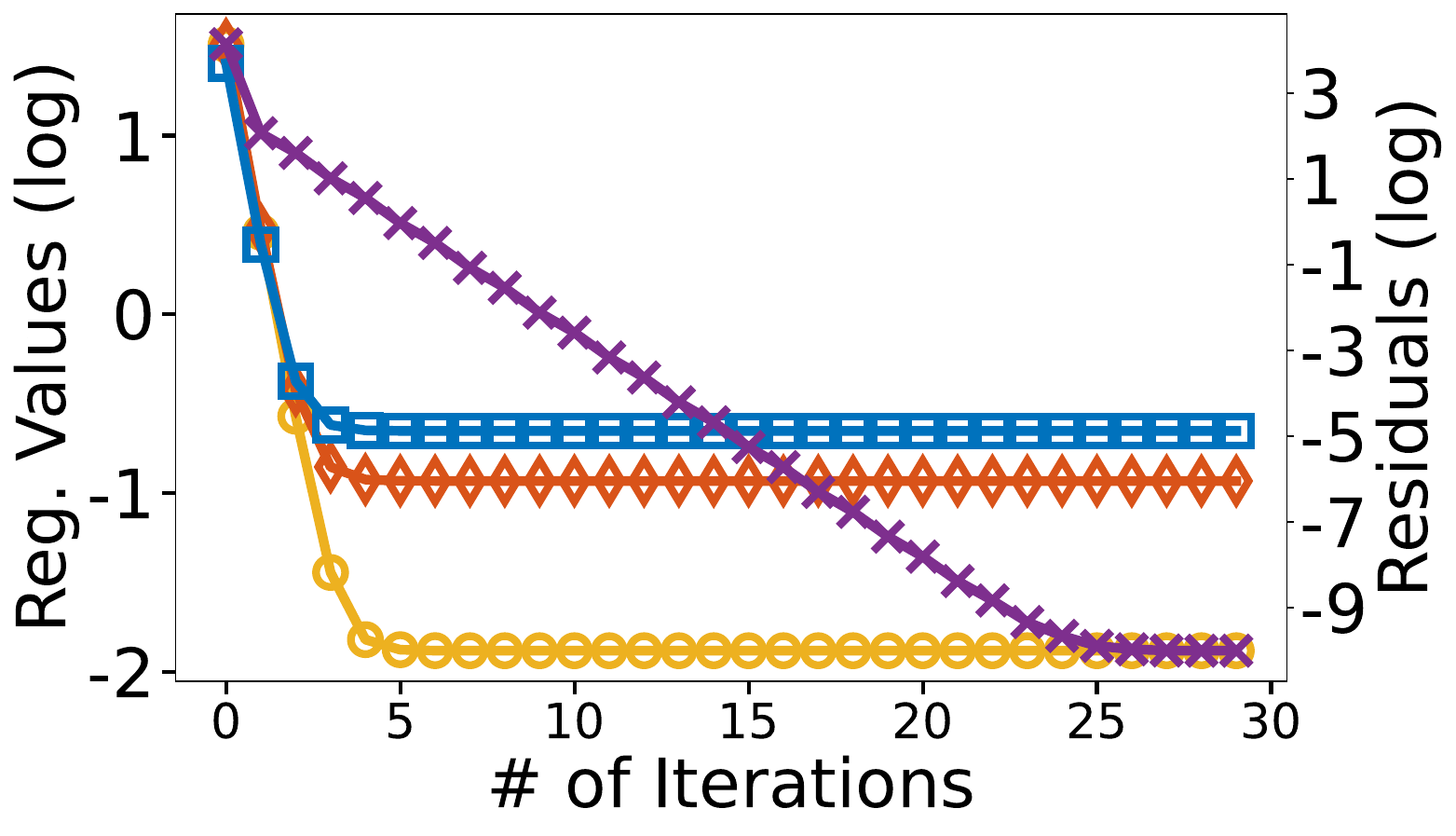}
        \hspace{-1mm}
    }
    \subfigure[\tmall]{
        \hspace{-3mm}
        \includegraphics[width=0.32\linewidth]{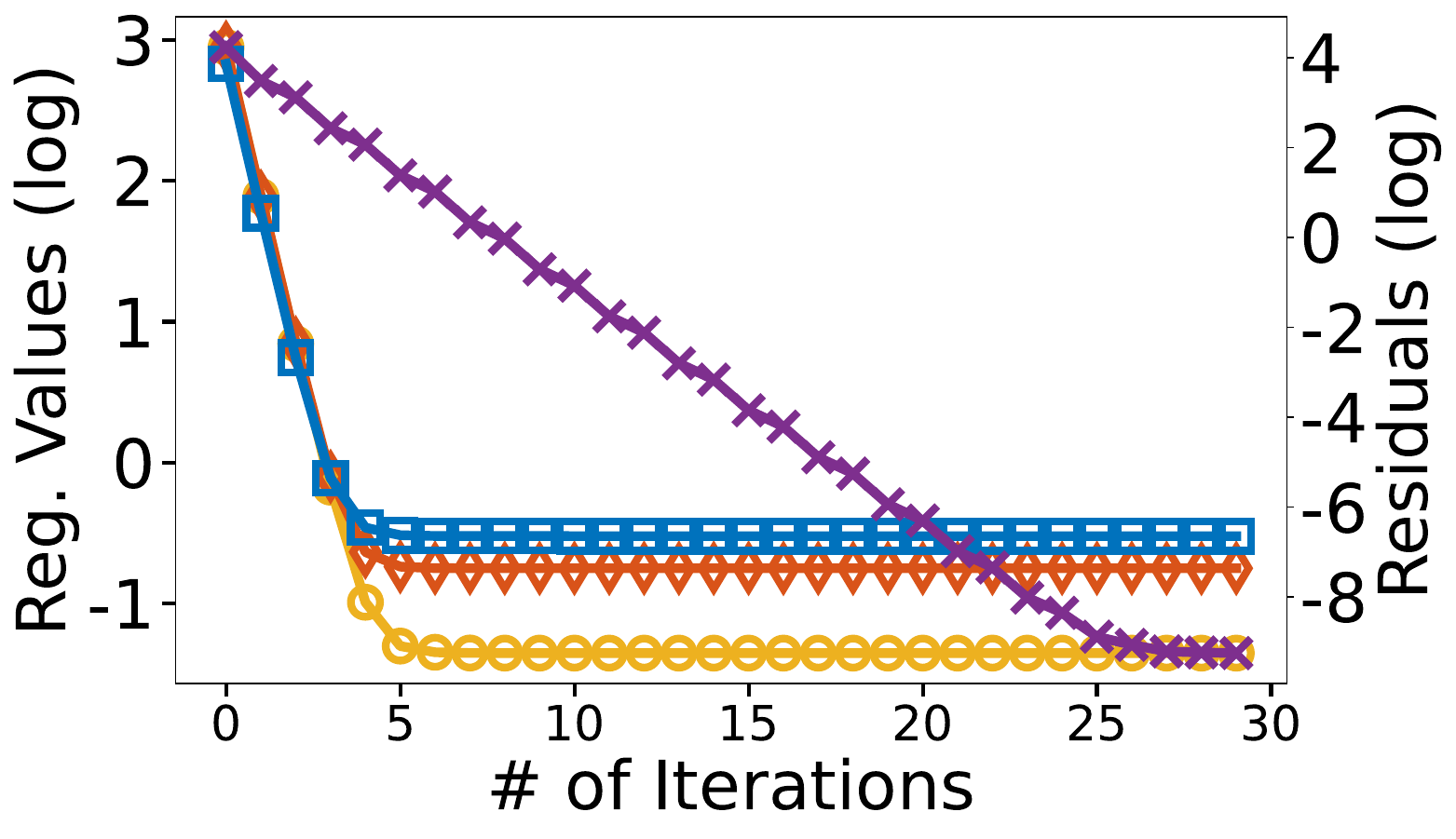}
        \hspace{-1mm}
    }
    
    \caption{
        \label{fig:experiments:reg_convergence}
        Convergence analysis on the regularization values of Equation~\eqref{eq:reg:obj} and the residuals in Algorithm~\ref{alg:method}.
        \textbf{As the number of iterations increases, the values of all regularization terms and residuals decrease and eventually converge.}
    }
\end{figure}

\begin{figure}[t]
    \centering
    \subfigure[\taobao]{
        \hspace{-5mm}
        \includegraphics[width=0.302\linewidth]{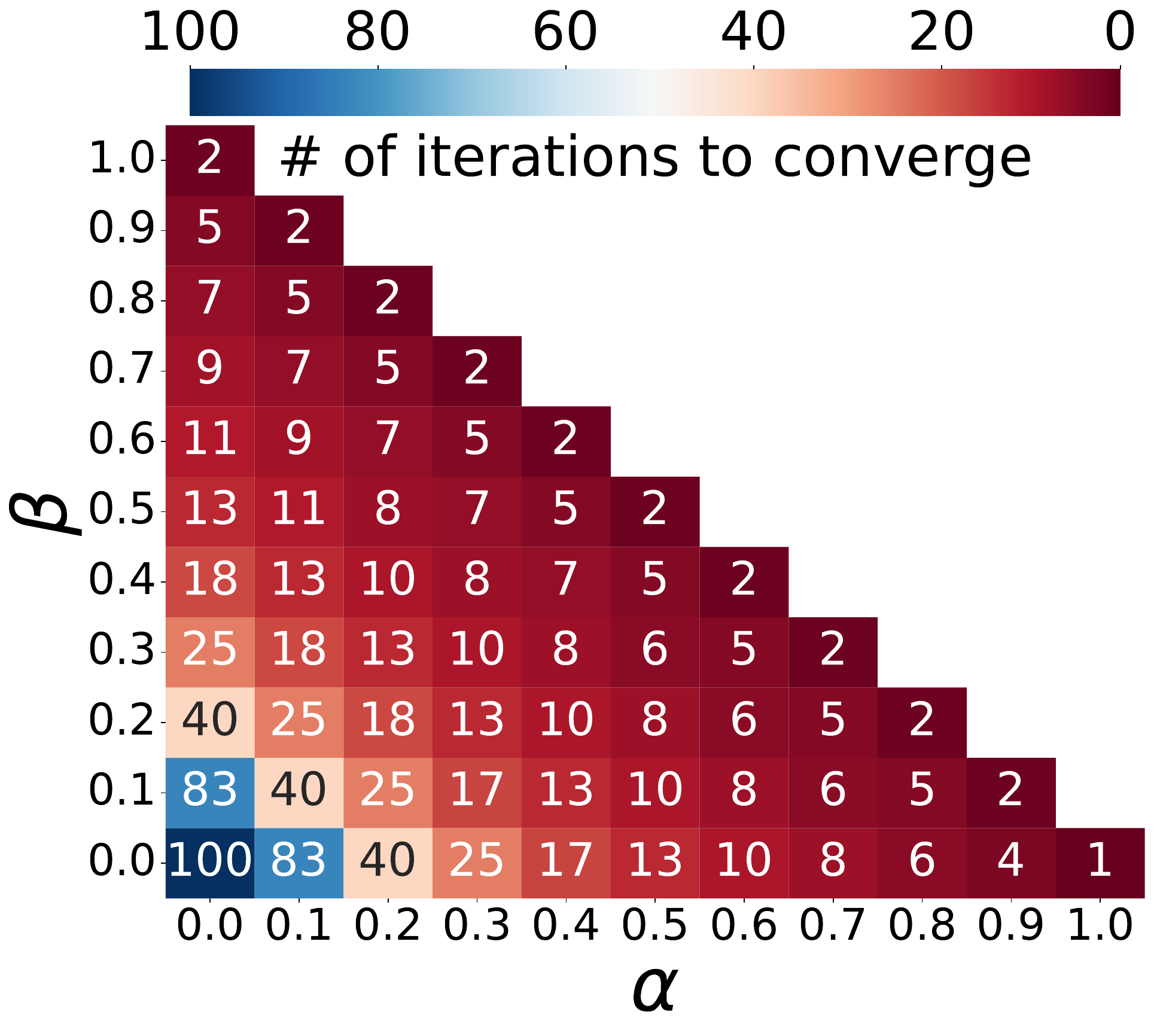}
        \hspace{-1mm}
        \label{fig:experiments:iteration:Taobao}
    }
    \subfigure[\tenrec]{
        \hspace{-3mm}
        \includegraphics[width=0.300\linewidth]{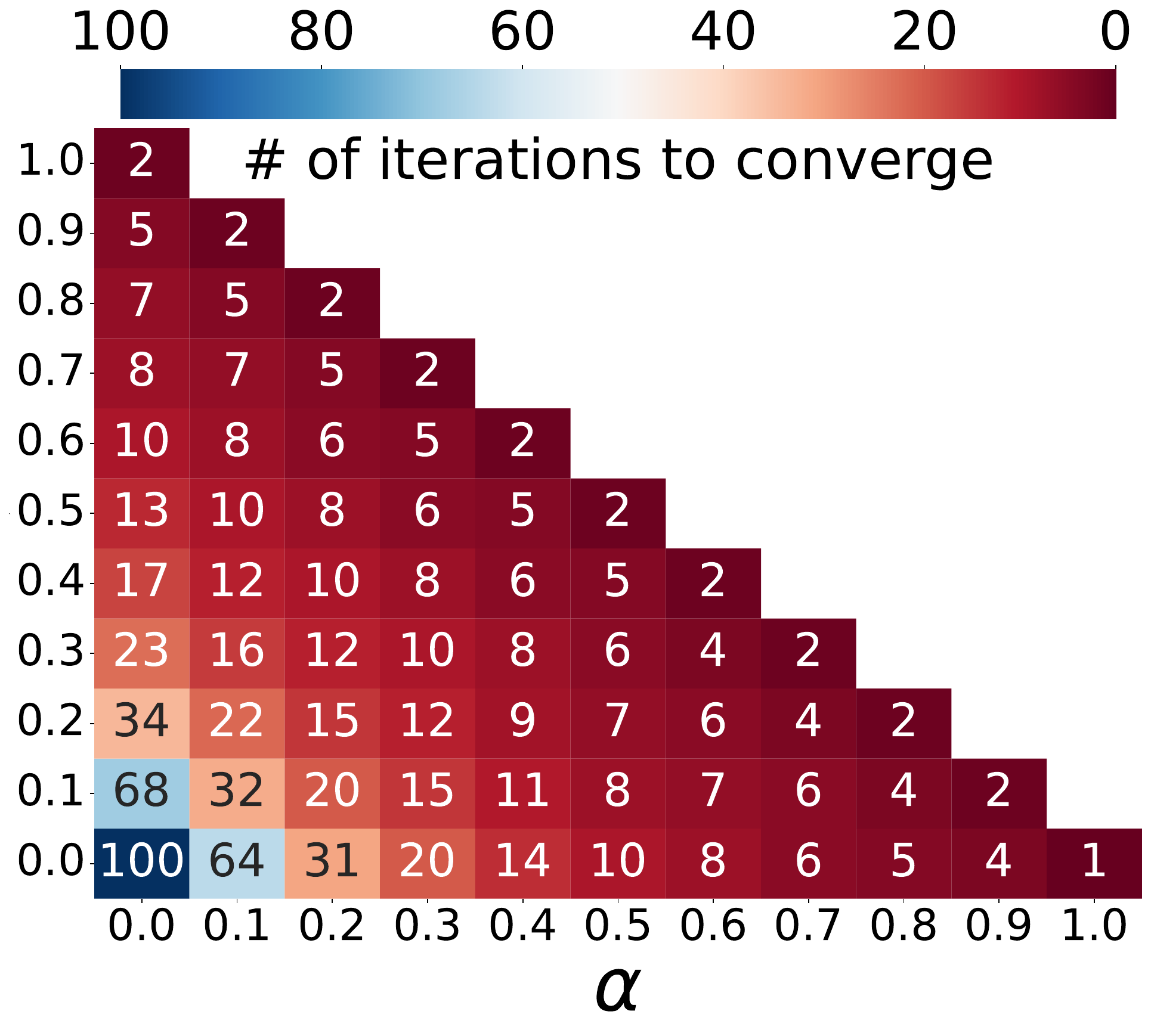}
        \hspace{-1mm}
        \label{fig:experiments:iteration:Tenrec}
    }
    \subfigure[\tmall]{
        \hspace{-3mm}
        \includegraphics[width=0.30\linewidth]{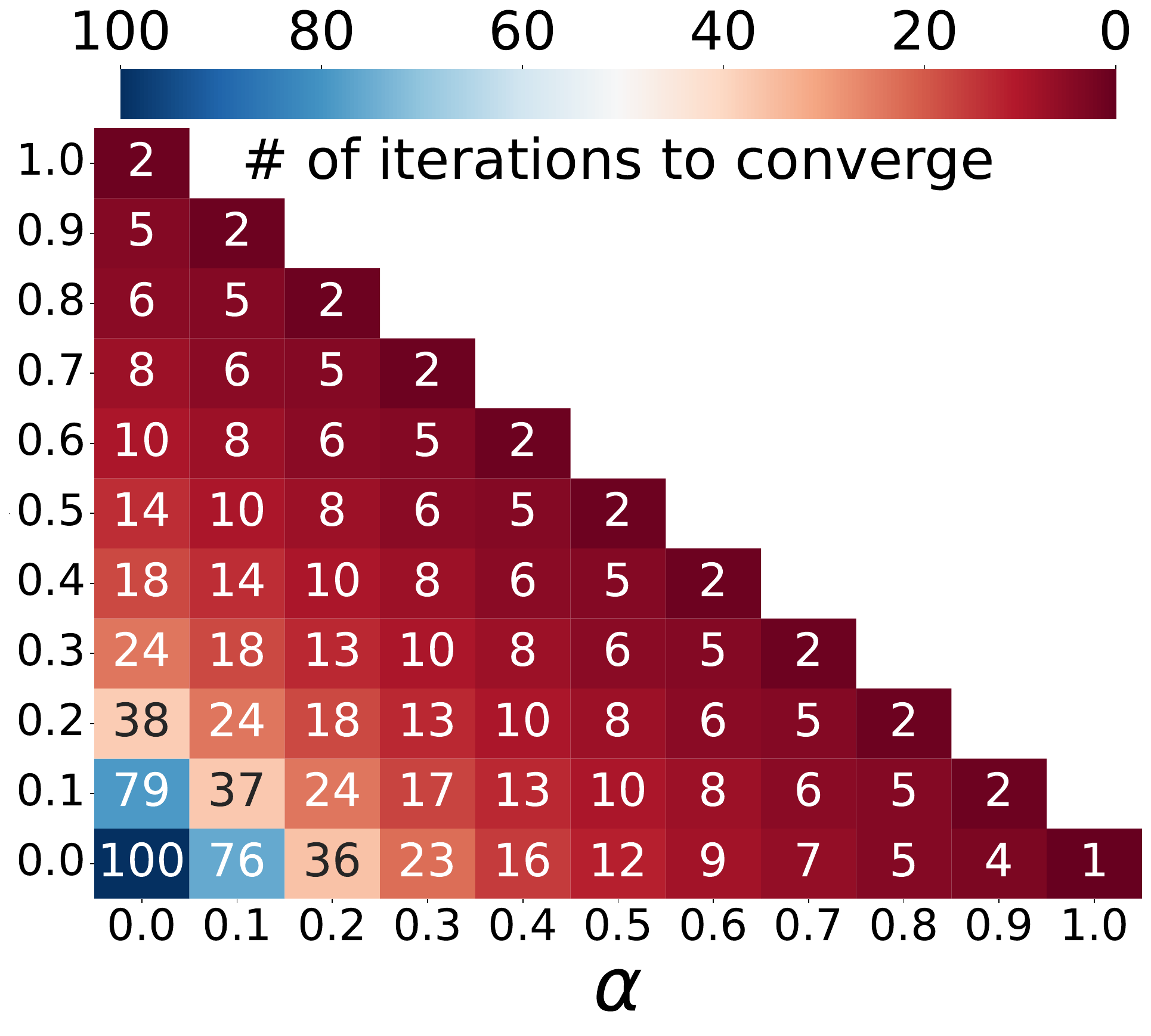}
        \hspace{-1mm}
        \label{fig:experiments:iteration:Tmall}
    }
    
    \caption{
        \label{fig:experiments:iteration}
        Effect of the hyperparameters $\alpha$ and $\beta$ of \method on the number of iterations to converge, where $\epsilon$ is set to $10^{-5}$.
        \textbf{Our algorithm for \method converges for all combinations of $\alpha$ and $\beta$ such that $\alpha + \beta \in (0, 1)$, with faster convergence for larger values of $\alpha + \beta$.}
    }
\end{figure}

\subsection{Convergence Analysis (Q4)}
\label{sec:experiments:convergence}
In this section, we analyze the convergence of the iterative algorithm for \method.

\smallsection{Analysis on regularization and residual}
We measured the average residuals of Algorithm~\ref{alg:method} and regularization values of Equation~\eqref{eq:reg:obj} for smoothing, query, and cascading across all querying users as the number of iterations increased.
To broadly observe the convergence of these terms, we set $\alpha = 0.3$ and $\beta = 0.4$ and randomly initialized the ranking vectors in Algorithm~\ref{alg:method}.
Figure~\ref{fig:experiments:reg_convergence} shows the results of this analysis, with the left $y$-axis representing the log values of regularizations and the right $y$-axis representing the log values of residuals.
The values of all regularization terms and residuals decrease and converge as the number of iterations increases sufficiently.
This indicates that Algorithm~\ref{alg:method} ensures convergence of the residuals, and the resulting scores minimize the objective function of Equation~\eqref{eq:reg:obj}, i.e., they adhere to ranking smoothness while aligning with the querying and cascading vectors. 
Note that convergence is guaranteed for any valid value of $\alpha$ and $\beta$, as discussed in Section~\ref{sec:proposed:analysis}.

\smallsection{Analysis on the number of iterations}
We further analyzed the number of iterations to convergence for various values of $\alpha$ and $\beta$, where the threshold $\epsilon$ for convergence is set to $10^{-5}$.
As shown in Figure~\ref{fig:experiments:iteration}, all valid combinations of $\alpha$ and $\beta$ where $\alpha + \beta \in (0, 1]$ result in convergence. 
Note that larger values of $\alpha + \beta$ lead to faster convergence because $\gamma = 1 - \alpha - \beta$ becomes smaller, which shrinks the range of the eigenvalues of $\mat{S}$.

\subsection{Computational Efficiency (Q5)}
\label{sec:exp:efficiency}
We evaluated the efficiency of our proposed \method in terms of scalability and the trade-off between  accuracy and running time. 

\begin{figure}[t]
    \centering  
    \vspace{3mm}
    \subfigure[\taobao]{
        \includegraphics[width=0.45\linewidth]{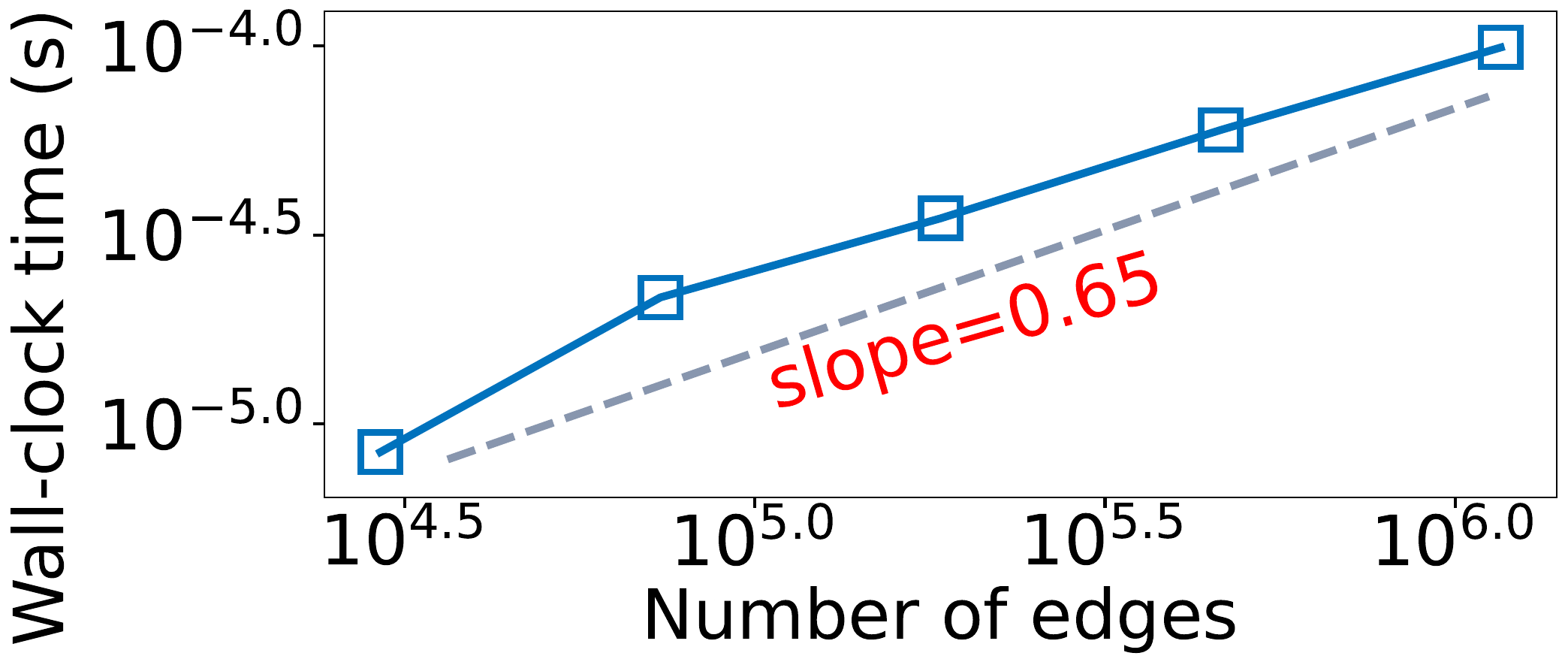}
    }
    \subfigure[\tenrec]{
        \includegraphics[width=0.4473447598\linewidth]{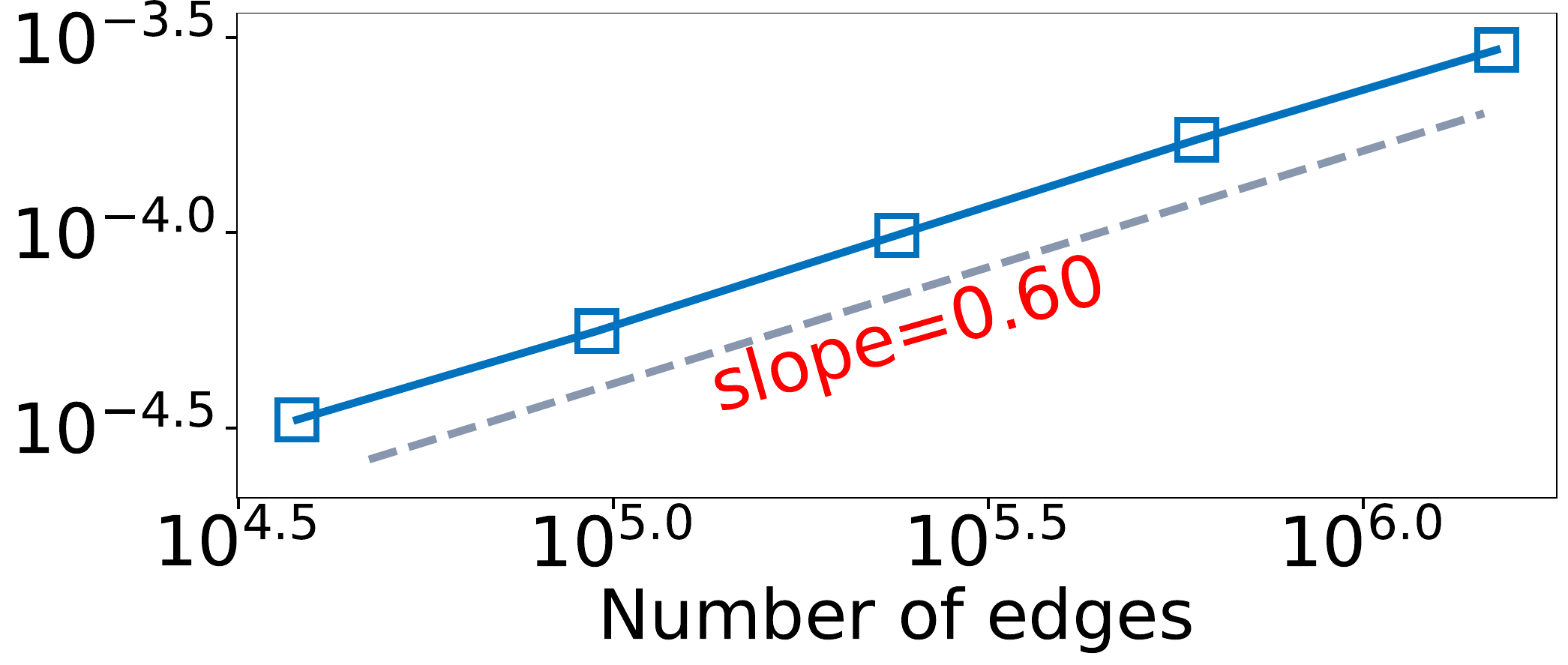}
    }
    \caption{
        \label{fig:scalability}
        Scalability of \method.
        \textbf{Our iterative algorithm for \method scales linearly with respect to the number of edges (or interactions).}
    }
\end{figure}

\begin{figure}[t!]
    \centering
    \hspace{5mm}\includegraphics[width=0.7\linewidth]{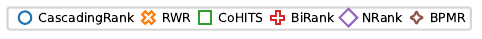}\vspace{-2mm}\\
    \subfigure[\taobao]{
        \hspace{-7mm}
        \includegraphics[width=0.3238536585\linewidth]{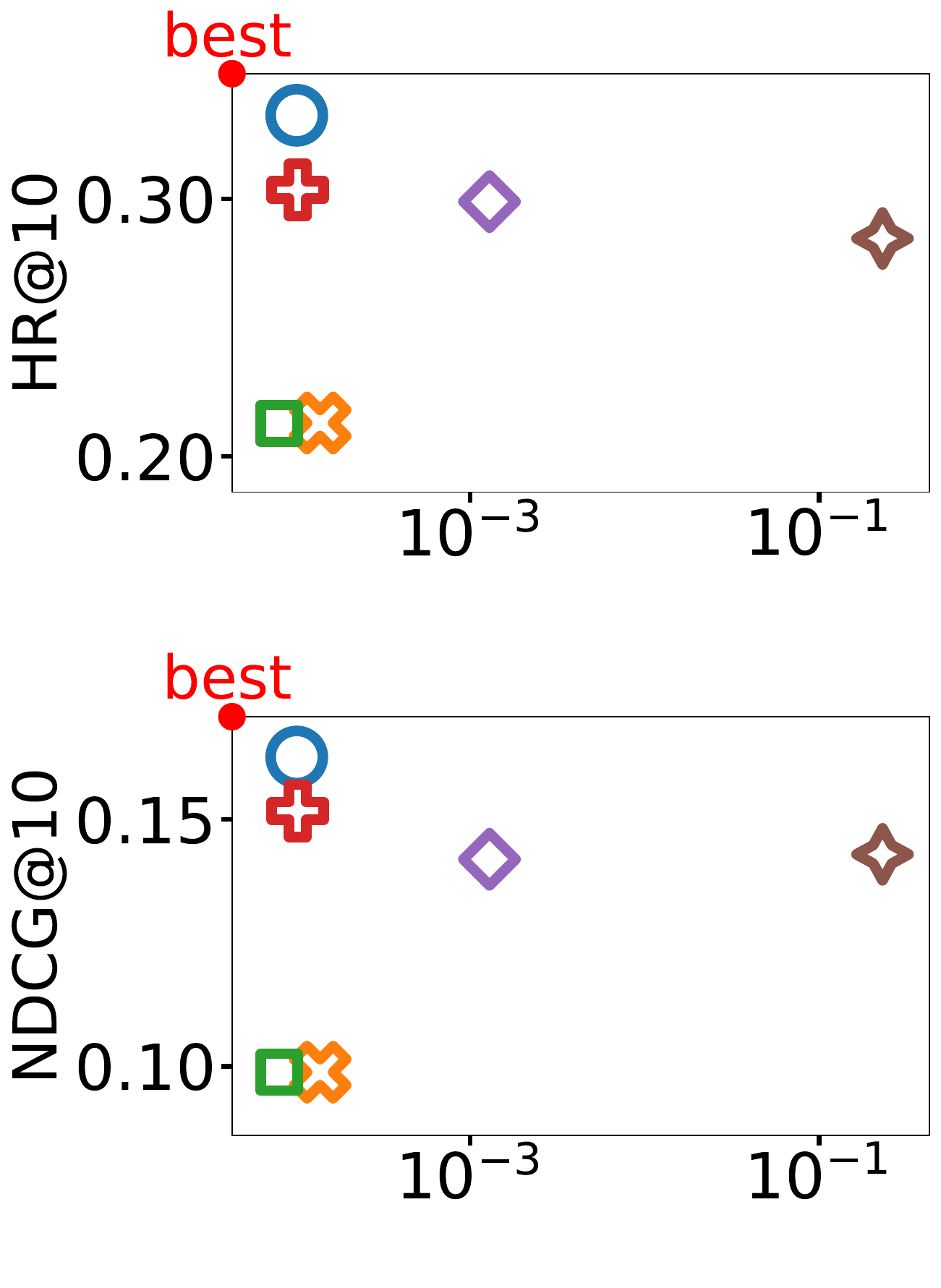}
        \hspace{-1mm}
    }
    \subfigure[\tenrec]{
        \hspace{-3mm}
        \includegraphics[width=0.300\linewidth]{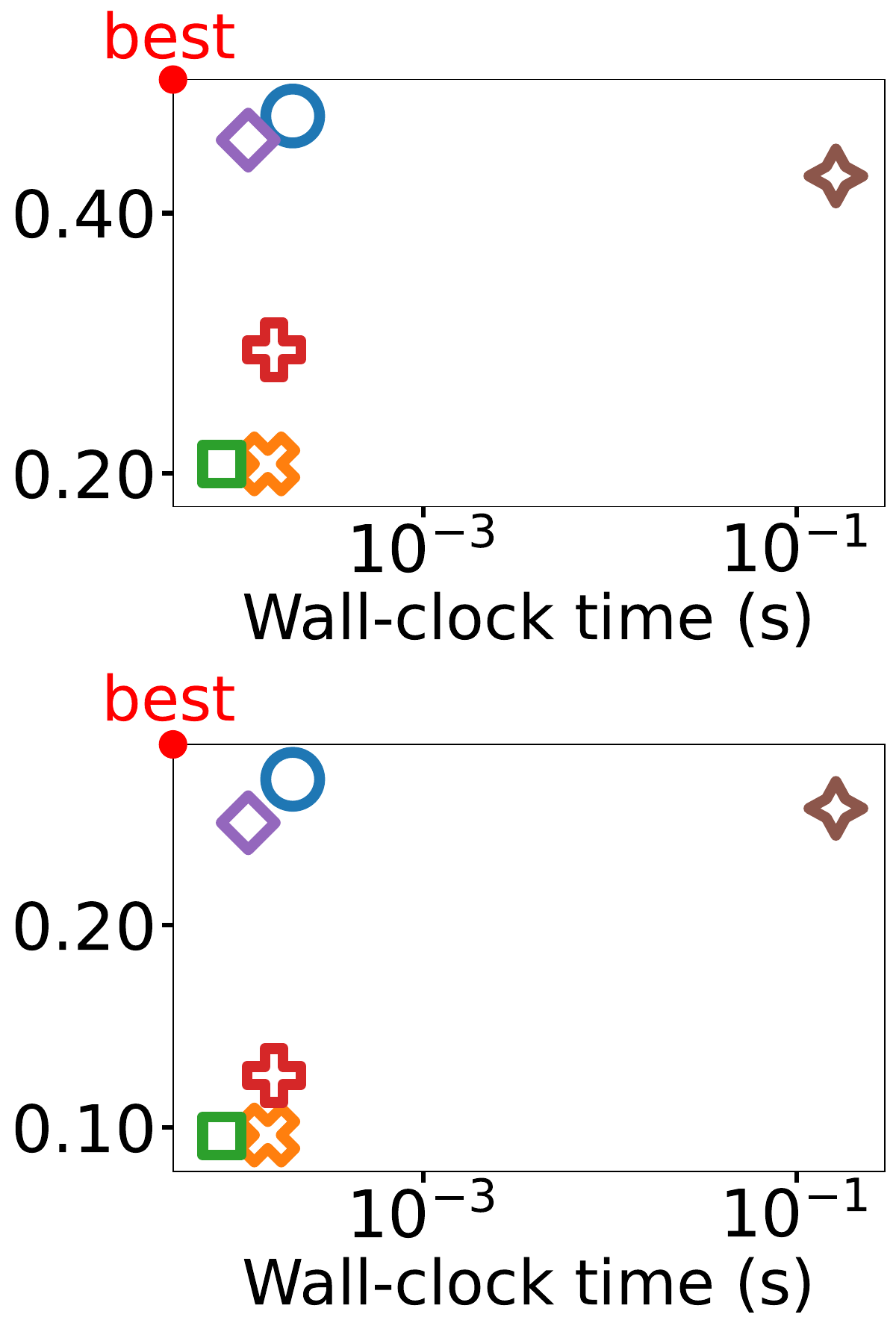}
        \hspace{-1mm}
        \label{fig:experiments:tradeoff:Tenrec}
    }
    \subfigure[\tmall]{
        \hspace{-3mm}
        \includegraphics[width=0.30\linewidth]{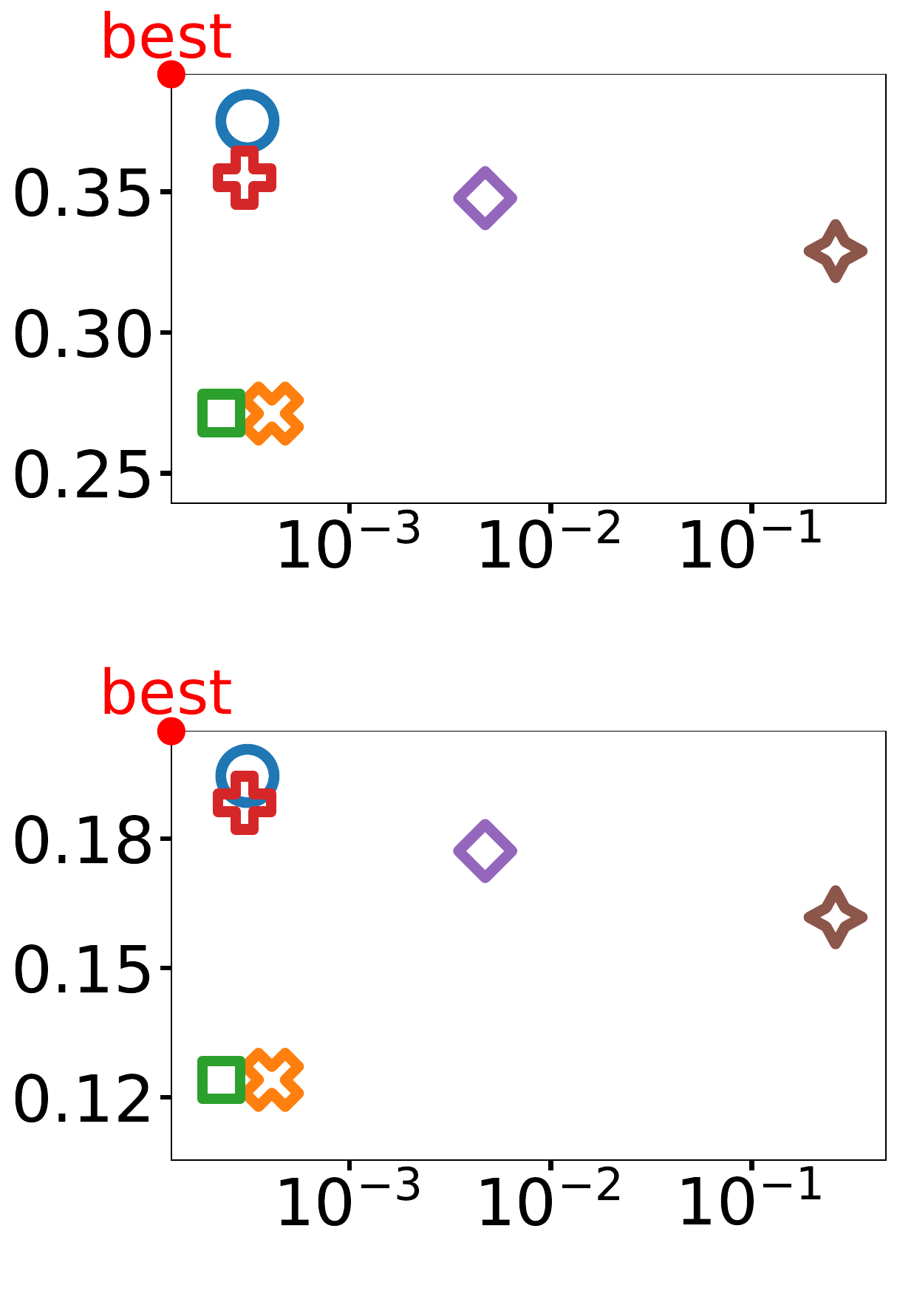}
        \hspace{-3mm}
        \label{fig:experiments:tradeoff:Tmall}
    }
    
    \caption{
        \label{fig:experiments:trade_off}
        Trade-off between accuracy and running time for the graph ranking methods.
        \textbf{Note that our \method achieves the highest accuracy while maintaining competitive runtime performance compared to other methods.}
    }
\end{figure}

\smallsection{Scalability}
To assess scalability, we measured the running time of \method by varying the number of interactions, using the \taobao and \tenrec datasets, which contain a large number of interactions. 
For each dataset, we first apply the same random permutation to all adjacency matrices, and then extract principal submatrices from each by slicing the upper-left part, ensuring that the number of interactions ranges from about $10^{4.5}$ to the original number of interactions.
Since bipartite graphs usually have different row and column sizes, we applied the same ratio (i.e., from 0 to 1) to both dimensions.
Figure \ref{fig:scalability} demonstrates that the running time of \method increases linearly with the number of interactions on both datasets, consistent with the theoretical analysis in Theorem~\ref{sec:proposed:analysis_complexity}.

\smallsection{Trade-off between accuracy and running time}
We also analyzed the trade-off between accuracy and running time. 
Note that \method is categorized as a graph ranking method. 
Therefore, we compared it with other graph ranking methods in this experiment (refer to Appendix~\ref{sec:appendix:EfficiencyComparisonwithRepresentationLearningMethods} for a comparison with representation learning methods).
Figure~\ref{fig:experiments:trade_off} demonstrates that our \method achieves the highest accuracy while maintaining competitive runtime performance compared to other graph ranking methods.
The traditional graph ranking methods, such as BiRank and NRank, demonstrate either competitive speed or accuracy compared to \method, but they fail to achieve both simultaneously.
Note that BPMR, a state-of-the-art method, performs worse than the graph ranking methods including \method in running time due to its sequential CPU-based algorithmic design (i.e., it is hard to parallelize), while others leverage parallelizable matrix operations on a GPU.

%% file: 060conclusion.tex
\section{Conclusion}
\label{sec:conclusion}
In this paper, we propose \method, a novel personalized graph ranking method for accurate multi-behavior recommendation. 
Our main idea is to leverage a cascading sequence of user behaviors leading to the target behavior by constructing a cascading behavior graph and measuring ranking scores along the graph, ensuring smoothness, query fitting, and cascading alignment.
We develop an iterative algorithm for scalable computation of ranking scores and theoretically analyze \method and its algorithm in terms of cascading effect, convergence, and scalability.
Through extensive experiments on three real-world multi-behavior datasets, we demonstrate that \method outperforms its competitors, particularly representation learning methods, and highlight the superiority of the graph ranking approach for multi-behavior recommendation.

%% file: 070appendix.tex
\newpage
\appendix
\section{Appendix}
\label{sec:appendix}

\subsection{Validated Hyperparameters of \method}
\label{sec:appendix:Effect of Hyperparameters_detail}
In Table~\ref{tab:appendix:besthyperparameter}, we report the validated values of the hyperparameters $\alpha$ and $\beta$ that provide the best performance on the validation set for each metric and dataset, where $\gamma$ is calculated as $\gamma = 1 - \alpha - \beta$.

\setlength{\tabcolsep}{15pt}
\begin{table}[h!]
\small
\caption{
    Validated values of the hyperparameters of \method, providing the best performance for each metric and dataset.
}
\label{tab:appendix:besthyperparameter}
\centering
\begin{tabular}{c|ccc|ccc}
\hline
\toprule
\textbf{Metric}  & \multicolumn{3}{c|}{\textbf{HR@10}} & \multicolumn{3}{c}{\textbf{NDCG@10}} \\
\midrule
\textbf{Datasets} & $\alpha$ & $\beta$ & $\gamma$ & $\alpha$ & $\beta$ & $\gamma$  \\
\midrule
\textbf{Taobao}  & 0              & 0.9   & 0.1         & 0               & 0.9 & 0.1             \\
\textbf{Tenrec}  & 0.3            & 0.6   & 0.1         & 0               & 0.6 & 0.4 \\
\textbf{Tmall}   & 0.7            & 0.2   & 0.1         & 0.7             & 0.2 & 0.1            \\
\bottomrule\hline
\end{tabular}
\end{table}

\subsection{Further Experiment on a Single Behavior Graph} 
\label{sec:appendix:single}

We additionally measured the performance of traditional methods in multi-behavior recommendation, which are designed for single-behavior graphs. 
Specifically, we used the interactions associated with the target behavior (e.g., \texttt{buy}) solely for their intended purpose.
As shown in Table~\ref{tab:appendix:SingleBehaviorCompareTable}, the performance is significantly low because the graph is extremely sparse when only single-behavior interactions are used.
In contrast, merging all interactions into a unified graph yields better performance than using a single-behavior graph, as it resolves the sparsity issue.
In Section~\ref{sec:experiments}, we report the unified graph's performance for each method as the baseline of \method.

\def\arraystretch{1.2} 
\setlength{\tabcolsep}{6pt}
\begin{table}[h]
\small
\caption{
Performance comparison of existing methods using either a single behavior graph or a unified behavior graph.
Relying on a single behavior results in poorer performance compared to utilizing all interactions across multiple behaviors.
\label{tab:appendix:SingleBehaviorCompareTable}
}
\centering
\begin{tabular}{cc|c|ccr|ccr}
\hline
\toprule
\multicolumn{1}{c}{\multirow{2}{*}{\textbf{Methods}}} & \multicolumn{1}{c|}{\multirow{2}{*}{\textbf{Type}}} & \multicolumn{1}{c|}{\multirow{2}{*}{\textbf{Datasets}}} & \multicolumn{3}{c|}{\textbf{HR@10}}                   & \multicolumn{3}{c}{\textbf{NDCG@10}}                  \\ 

\multicolumn{1}{l}{} & \multicolumn{1}{l|}{} & \multicolumn{1}{l|}{}  & \textbf{Single} & \textbf{Unified} & \textbf{\% impv.} & \textbf{Single} & \textbf{Unified} & \textbf{\% impv.} \\ \midrule

\multirow{3}{*}{\textbf{MF-BPR}}   &  \multirow{3}{*}{RL}          & \taobao                                        & 0.0371          & 0.0758           & 104.2\%            & 0.0177          & 0.0387           & 118.4\%            \\
&   & \tenrec                                        & 0.0092          & 0.1244           & 1249.5\%           & 0.0042          & 0.0575           & 1275.6\%           \\
&   & \tmall                                         & 0.0647          & 0.0855           & 32.0\%             & 0.0291          & 0.0423           & 45.4\%             \\ \midrule
\multirow{3}{*}{\textbf{LightGCN}}    &  \multirow{3}{*}{RL}                & \taobao                                        & 0.0368          & 0.1025           & 178.8\%            & 0.0216          & 0.0566           & 162.3\%            \\
 &    & \tenrec                                        & 0.0046          & 0.1069           & 2224.3\%           & 0.0023          & 0.0526           & 2147.9\%           \\
  &   & \tmall                                         & 0.0913          & 0.1162           & 27.3\%             & 0.0402          & 0.0625           & 55.3\%             \\ \midrule
\multirow{3}{*}{\textbf{BiRank}}       &  \multirow{3}{*}{GR}               & \taobao                                        & 0.0390          & 0.3034           & 678.6\%            & 0.0216          & 0.1517           & 602.2\%            \\
  &   & \tenrec                                        & 0.0046          & 0.2949           & 6300.0\%           & 0.0018          & 0.1257           & 6951.5\%           \\
 &    & \tmall                                         & 0.0849          & 0.3550           & 317.9\%            & 0.0519          & 0.1819           & 250.4\%            \\ \midrule
\multirow{3}{*}{\textbf{CoHITS}}          &  \multirow{3}{*}{GR}    & \taobao      & 0.0404          & 0.2128           & 426.9\%            & 0.0241          & 0.0988           & 310.1\%            \\
  &   & \tenrec                                        & 0.0046          & 0.2074           & 4400.0\%           & 0.0020          & 0.0957           & 4722.6\%           \\
  &   & \tmall                                         & 0.0735          & 0.2713           & 269.2\%            & 0.0458          & 0.1284           & 180.2\%            \\ \midrule
\multirow{3}{*}{\textbf{RWR}}               &  \multirow{3}{*}{GR}   & \taobao                                        & 0.0412          & 0.2130           & 417.5\%            & 0.0246          & 0.0988           & 301.4\%            \\
 &    & \tenrec                                        & 0.0000          & 0.2074           & $\infty$              & 0.0000          & 0.0962           & $\infty$              \\
  &   & \tmall                                         & 0.0726          & 0.2712           & 273.8\%            & 0.0452          & 0.1284           & 184.1\%            \\ \bottomrule \hline 
\end{tabular}
\end{table}

\newpage
\subsection{Efficiency Comparison with Representation Learning Methods}
\label{sec:appendix:EfficiencyComparisonwithRepresentationLearningMethods}

We compared the graph ranking (GR) methods including \method with the representation learning (RL) and pattern mining (PM) methods in terms of running time and accuracy. 
For a fair comparison, we measured the wall-clock time (in seconds) required to generate ranking (or recommendation) scores for all users from the input graph. 
We applied early stopping to the RL methods, mirroring the convergence criteria used for the GR methods.
As shown in Table~\ref{tab:appendix:efficienycompare}, our \method provides the best accuracy while showing competitive speed compared to its competitors. 
Notably, the graph ranking approach, including our proposed \method, is significantly faster than the RL and PM approaches while also achieving higher accuracy, underscoring its strengths in both speed and accuracy.

\def\arraystretch{1.2} 
\setlength{\tabcolsep}{6pt}
\begin{table}[h]
\small
\centering
\caption{
\label{tab:appendix:efficienycompare}
Comparison of \method with representation learning (RL), graph ranking (GR), and pattern mining (PM) methods in terms of efficiency (Time) and accuracy (HR@10), where we measured the wall-clock time in seconds for the end-to-end process of generating recommendation scores for all users from the input graph.
Note that our \method is significantly faster than the RL and PM methods, and comparable to the GR methods, while achieving the best accuracy.
}
\begin{tabular}{cc|rc|rc|rc}
\hline
\toprule
\multicolumn{2}{c|}{\bf Datasets}       & \multicolumn{2}{c|}{\bf \taobao} & \multicolumn{2}{c|}{\bf \tenrec}  & \multicolumn{2}{c}{\bf \tmall}   \\
\textbf{Methods} & \textbf{Type}       & \textbf{Time (s) } & \textbf{HR@10}  & \textbf{Time (s)} & \textbf{HR@10}   & \textbf{Time (s)} & \textbf{HR@10}  \\ \midrule
MB-HGCN & RL  & 96.0               & 0.1261 & 126.4              & 0.1413 & 306.8              & 0.1133 \\
MuLe   & RL       & 764.4              & 0.1949 & 39.7             & 0.2097  & 3937.2               & 0.1920 \\
PKEF   & RL       & 2031.6             & 0.1349 & 2258.9             & 0.1222  & 1925.2             & 0.0968 \\
HEC-GCN   & RL    & 104.3              & 0.1905 & 134.6              & 0.1806  & 323.9              & 0.2673 \\ \midrule
RWR  & GR   & 2.1                & 0.2130 & 6.1               & 0.2712  & 11.4               & 0.2074 \\
CoHITS & GR  & 1.2               & 0.2128 & 3.5                & 0.2713  & 6.5                & 0.2074 \\
BiRank & GR  & 1.6                & 0.3034 & 6.6                & 0.3550  & 8.4                & 0.2949 \\
NRank  & GR        & 19.9               & 0.2989 & 4.8              & 0.3477  & 132.0              & 0.4562 \\
\midrule
BPMR  & PM        & 3567.9             & 0.2846 & 6753.1             & 0.3289  & 7300.2             & 0.4286 \\
\midrule
\bf \method & GR & 1.6                & 0.3324 & 8.3                & 0.3751  & 8.6                 & 0.4608 \\ \bottomrule \hline 
\end{tabular}

\end{table}

\subsection{Comparison of Different Cascading Sequences}
\label{sec:appendix:Performance_for_Different_Cascading_Sequences}

As described in Section~\ref{sec:exp:setting}, we fixed the cascading sequence $\C$ to the bold sequence shown in Table~\ref{tab:appendix:behaviorpertubation} for each dataset, assuming a natural sequence of user behaviors inspired by earlier works~\cite{LiuXWY00024, yin2024hecgcn, ChengCHLZGP23fqvn, YanCGSLSL24}.
However, the performance of \method can depend on the order of behaviors in $\C$, and the assumed sequence may therefore be suboptimal.
To examine the effect of the order of $\C$, we conducted an experiment in which the last behavior in $\C$ was fixed to the target behavior $b_t$, while the other behaviors were permuted.
Table~\ref{tab:appendix:behaviorpertubation} presents the experimental results in terms of HR@10 and NDCG@10.
Notably, in \taobao and \tmall, the sequence assumed in Section~\ref{sec:experiments} demonstrates suboptimal accuracy compared to other sequences, whereas it achieves the best performance in \tenrec.
Nevertheless, the assumed sequence for each dataset delivers competitive performance relative to the other sequences, outperforming the competitors of \method.
This result suggests that our cascading approach is effective in providing accurate recommendation, but the optimal sequence $\C$ of user behaviors depends on datasets, indicating that learning such a sequence in this setting is a promising direction of future work.

\setlength{\tabcolsep}{17pt}
\begin{table}[h]
\small
\caption{
\label{tab:appendix:behaviorpertubation}
Performance of \method for different cascading sequences in terms of HR@10 and NDCG@10. 
The last behavior is fixed to the target behavior, while the others are permuted in the sequence, with the sequence used in Section~\ref{sec:experiments} highlighted in bold.
}
\centering
\begin{tabular}{c|l|cc}
\hline
\toprule
\textbf{Datasets}                 & \textbf{Cascading Sequences}                  & \multicolumn{1}{l}{\textbf{HR@10}} & \multicolumn{1}{l}{\textbf{NDCG@10}} \\ \midrule
\multirow{2}{*}{\taobao} & \bf{view$\rightarrow$cart$\rightarrow$buy}           & 0.3324                    & 0.1626                      \\
                        & cart$\rightarrow$view$\rightarrow$buy           & \bf 0.3409                    & \bf 0.1675                      \\ \midrule
\multirow{6}{*}{\tenrec} & \bf{view$\rightarrow$share$\rightarrow$like$\rightarrow$follow} & \bf 0.4793                    & \bf 0.2723                      \\
                        & view$\rightarrow$like$\rightarrow$share$\rightarrow$follow & 0.4747                    & 0.2700                      \\
                        & share$\rightarrow$view$\rightarrow$like$\rightarrow$follow & 0.4700                    & 0.2545                      \\
                        & share$\rightarrow$like$\rightarrow$view$\rightarrow$follow & 0.4654                    & 0.2544                      \\
                        & like$\rightarrow$view$\rightarrow$share$\rightarrow$follow & \bf 0.4793                    & 0.2698                      \\
                        & like$\rightarrow$share$\rightarrow$view$\rightarrow$follow & 0.4700                    & 0.2578                      \\ \midrule
\multirow{6}{*}{\tmall}  & \bf{view$\rightarrow$collect$\rightarrow$cart$\rightarrow$buy}  & 0.3751                    & 0.1871                      \\
                        & view$\rightarrow$cart$\rightarrow$collect$\rightarrow$buy  & 0.3699                    & 0.1849                      \\
                        & collect$\rightarrow$view$\rightarrow$cart$\rightarrow$buy  & 0.3954                    & 0.1999                      \\
                        & collect$\rightarrow$cart$\rightarrow$view$\rightarrow$buy  & 0.3968                    & \bf 0.2086                      \\
                        & cart$\rightarrow$view$\rightarrow$collect$\rightarrow$buy  & 0.3757                    & 0.1866                      \\
                        & cart$\rightarrow$collect$\rightarrow$view$\rightarrow$buy  & \bf 0.3974                    & 0.2037                      \\ 
                        \bottomrule \hline
\end{tabular}
\end{table}

\newpage
\subsection{Detailed Analysis on Effect of Hyperparameters}
\label{sec:appendix:details:hyperparams}

We report how the hyperparameters $\alpha$ and $\beta$ affect \method across all possible combinations in terms of HR@10 and NDCG@10. 
Figure~\ref{fig:appendix:Hyper-sens-detail} displays the experimental results, showing that the trends vary across the datasets.
In \taobao, smaller values of $\alpha$ and higher values of $\beta$ tend to yield better performance, while in \tenrec, moderate $\beta$ values perform better.
In \tmall, higher values of $\alpha+\beta$ tend to show higher accuracy.

\begin{figure}[h]
    \centering
    \subfigure[\taobao]{
        \hspace{-5mm}
        \includegraphics[width=0.321\linewidth]{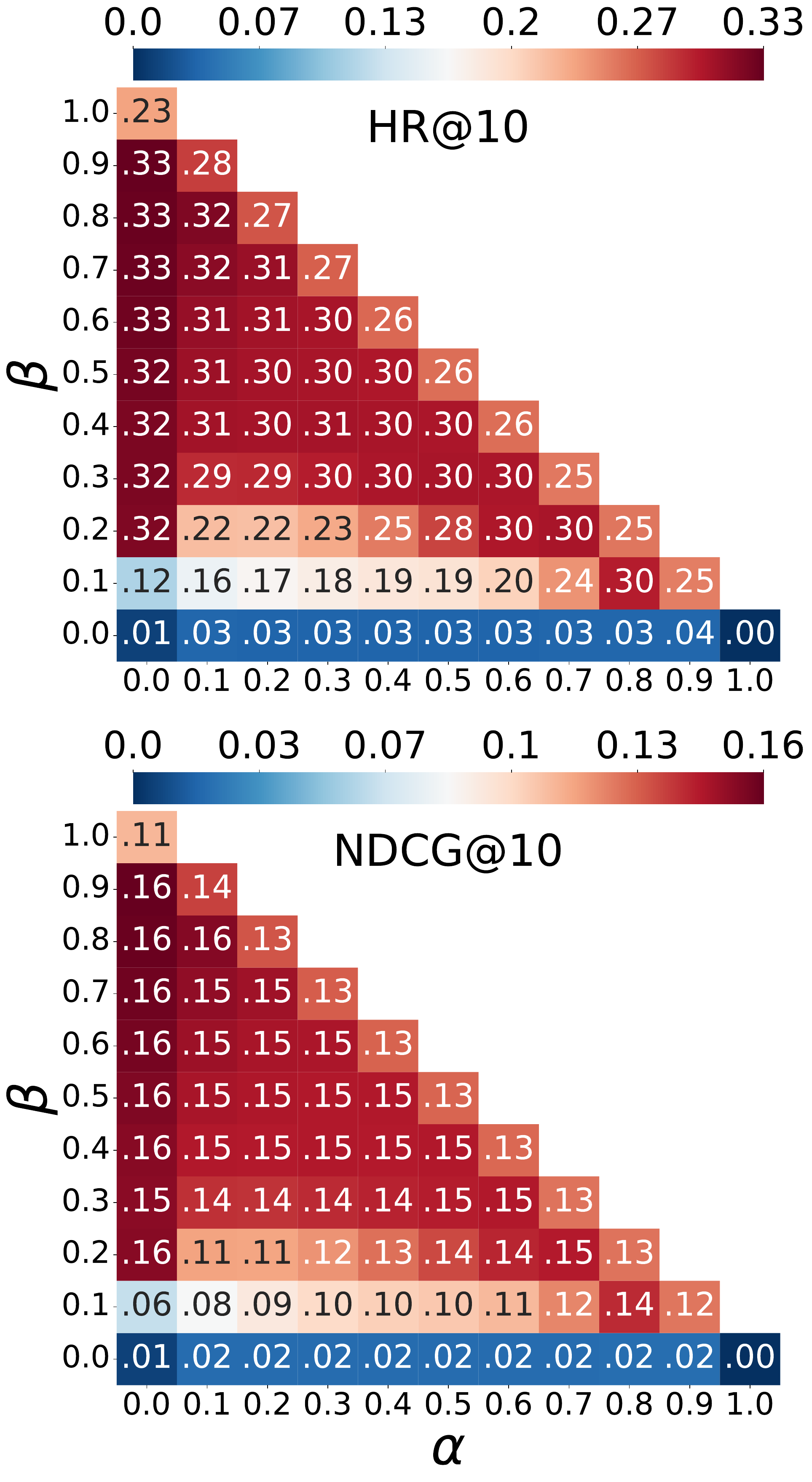}
        \hspace{-1mm}
        \label{fig:experiments:Hyper-sens:Taobao}
    }
    \subfigure[\tenrec]{
        \hspace{-3mm}
        \includegraphics[width=0.320\linewidth]{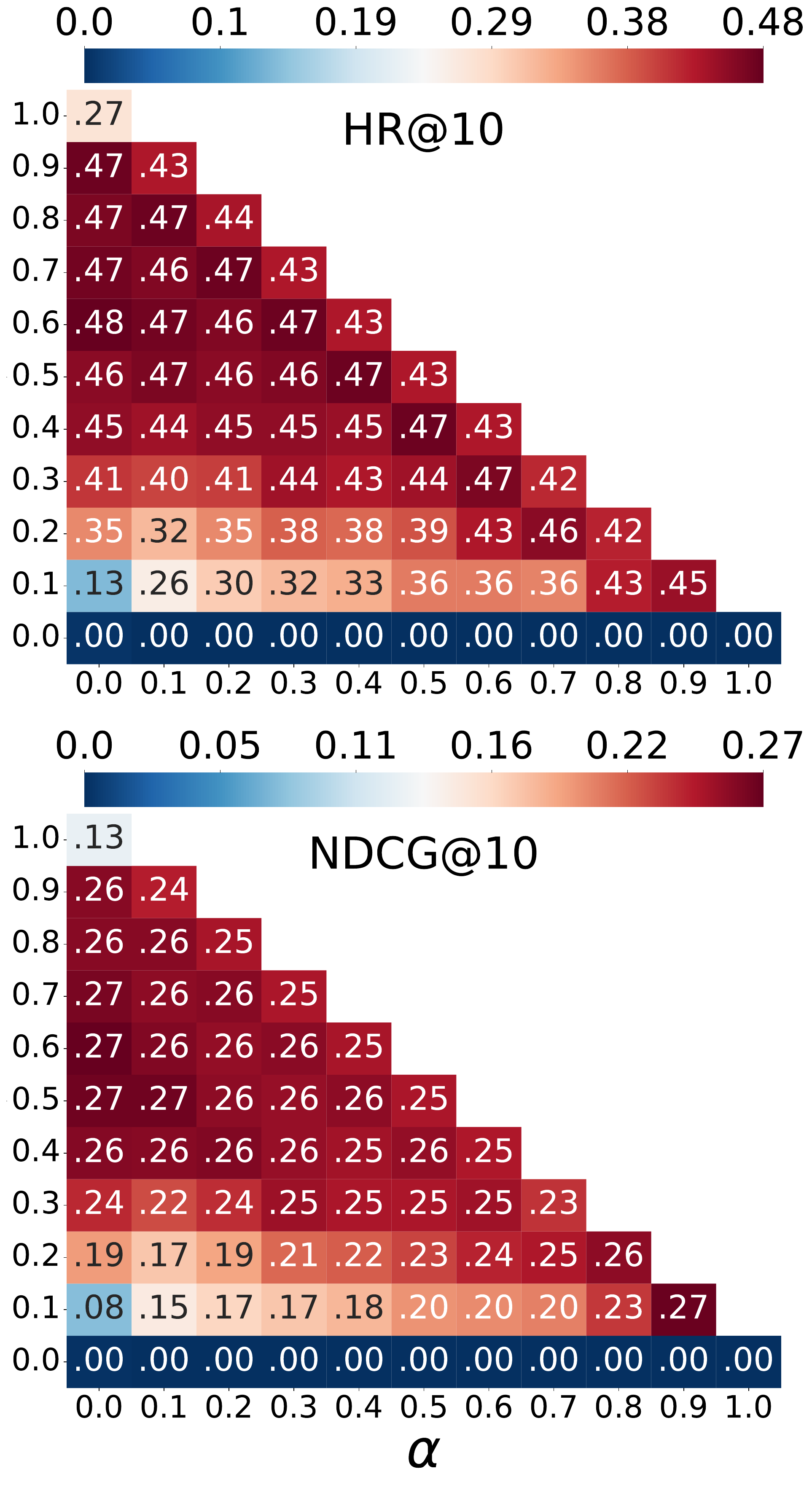}
        \hspace{-1mm}
        \label{fig:experiments:Hyper-sens:Tenrec}
    }
    \subfigure[\tmall]{
        \hspace{-3mm}
        \includegraphics[width=0.32\linewidth]{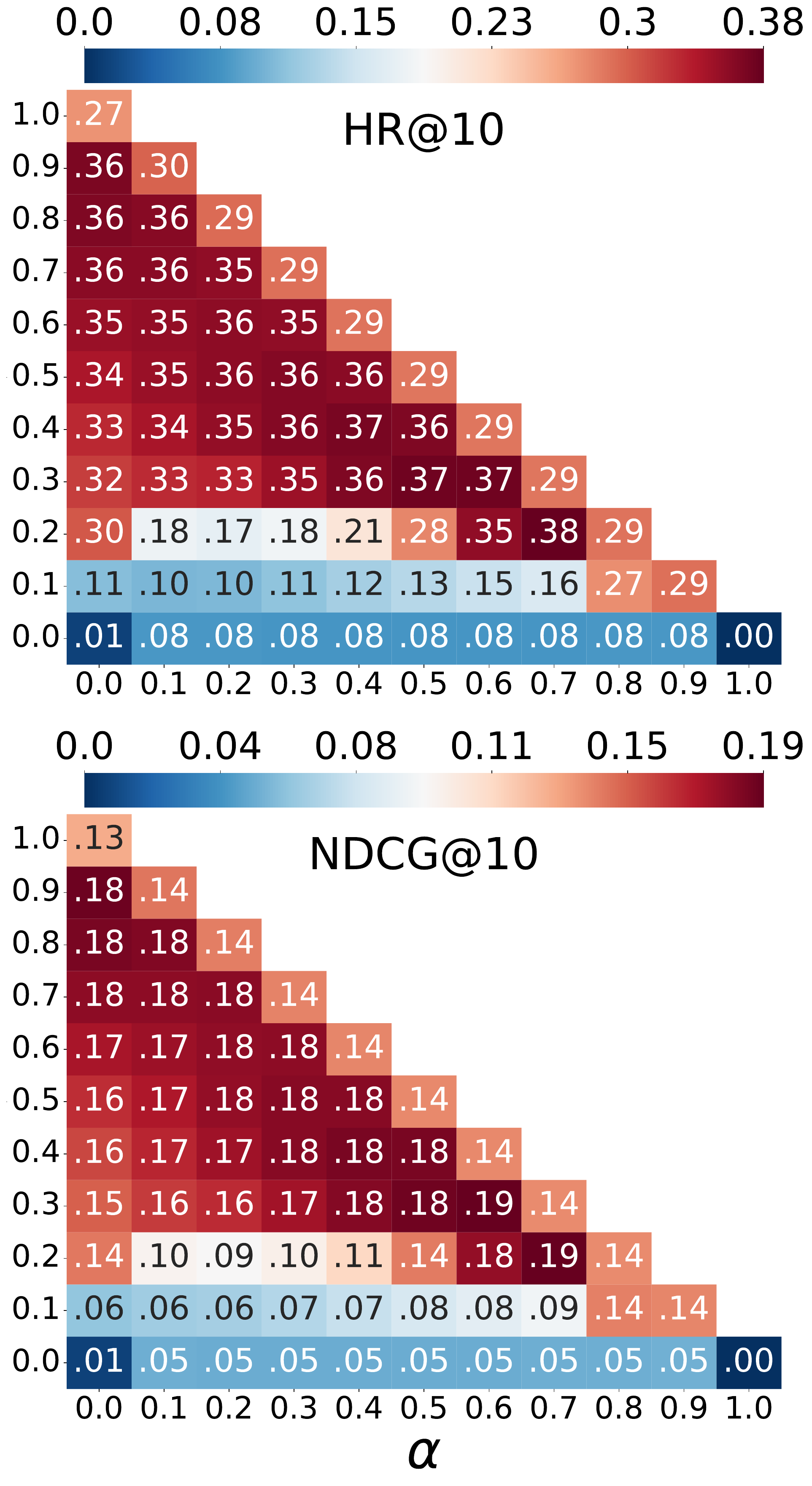}
        \hspace{-1mm}
        \label{fig:experiments:Hyper-sens:Tmall}
    }

    \caption{
        \label{fig:appendix:Hyper-sens-detail}
        Detailed effect of hyperparameters in \method, where $\alpha$ controls the strength of query fitting, and $\beta$ controls the strength of cascading fitting. These hyperparameters are searched within the range $0 \leq \alpha + \beta \leq 1$ and $0 \leq \alpha, \beta \leq 1$. 
    }
\end{figure}

\subsection{Implementation Information of Competitors}
\label{sec:appendix:competitors:information}
We used open-source implementations for the competitors in our experiments, with detailed information provided below:
\begin{itemize}[leftmargin=9mm,noitemsep]
    \item{ 
        \textbf{MF-BPR}: \url{https://github.com/RUCAIBox/RecBole} 
    }
    \item{
        \textbf{LightGCN}: \url{        https://github.com/RUCAIBox/RecBole}
    }
    \item{
        \textbf{MB-HGCN}: 
        \url{https://github.com/MingshiYan/MB-HGCN}
    }
    \item{
        \textbf{MuLe}: \url{https://github.com/geonwooko/MULE}
    }
    \item{
        \textbf{PKEF}: \url{https://github.com/MC-CV/PKEF}
    }
    \item{
        \textbf{HEC-GCN}: \url{https://github.com/marqu22/HEC-GCN}
    }
    \item{
        \textbf{RWR}: \url{https://github.com/jinhongjung/pyrwr}
    }
    \item{
        \textbf{CoHITS}: \url{https://github.com/BrianAronson/birankr}
    }
    \item{
        \textbf{BiRank}: \url{https://github.com/BrianAronson/birankr}
    }
    \item{
        \textbf{NRank}: \url{https://github.com/BrianAronson/birankr}
    }
    \item{
        \textbf{BPMR}: \url{https://github.com/rookitkitlee/bpmr}
    }
\end{itemize}